\documentclass[11pt]{article}

\usepackage{amsfonts}

\usepackage{multirow}
\usepackage{amsmath, dsfont, amscd, latexsym,  color, amsthm}
\usepackage{amssymb}
\usepackage{ifthen}
\usepackage{enumerate}

\usepackage{url}
\newtheorem{theorem}{Theorem}

\newtheorem{lemma}[theorem]{Lemma}

\newtheorem{proposition}[theorem]{Proposition}
\newtheorem{corollary}[theorem]{Corollary}

\newtheorem{myclaim}[theorem]{Claim}
\newtheorem{fact}[theorem]{Fact}

\newtheorem{definition}[theorem]{Definition}

\newcommand{\braket}[1]{\langle #1 \rangle}

\newcommand{\weight}{{\mathrm{wt}}}
\newcommand{\dshap}{{d_{\mathrm{Shapley}}}}
\newcommand{\dfour}{{d_{\mathrm{Fourier}}}}

\newcommand{\eps}{\epsilon}

\newcommand{\sgn}{\mathrm{sign}}
\newcommand{\sign}{\mathrm{sign}}
\newcommand{\mytextsf}[1]{{{\textsf {#1}}}}

\newcommand{\ignore}[1]{}

\newcommand{\cref}[1]{Corollary~\ref{cor:#1}}

\newcommand{\bits}{\{-1,1\}}
\newcommand{\bn}{\bits^n}

\newcommand{\R}{{\mathbb{R}}}

\newcommand{\Z}{{\mathbb Z}}
\newcommand{\E}{\operatorname{{\bf E}}}

\newcommand{\littlesum}{\mathop{\textstyle \sum}}

\newcommand{\poly}{\mathrm{poly}}

\newcommand{\eqdef}{\stackrel{\textrm{def}}{=}}

\newcommand{\xop}{x(\pi,i)}

\newcommand{\xip}{x^+(\pi,i)}

\newcommand{\mfed}{ \mathbf{E}_{x \sim \mu} }
\renewcommand{\Pr}{\operatorname{{\bf Pr}}}

\newcommand{\ms}{\mathbb{S}_n}

\newcommand{\rnote}[1]{\footnote{{\bf [[Rocco: {#1}\bf ]] }}}

\newcommand{\new}[1]{{{#1}}}
\newcommand{\newer}[1]{{{#1}}}

\begin{document}

%\mainmatter  % start of an individual contribution

\title{The Inverse Shapley Value Problem\footnote{An extended abstract of this work appeared in the {\em Proceedings of the 39th International Colloquium on Automata, Languages and Programming (ICALP 2012).}}}

\title{The Inverse Shapley Value Problem\footnote{An extended abstract of this work appeared in the {\em Proceedings of the 39th International Colloquium on Automata, Languages and Programming (ICALP 2012).}}}

\author{Anindya De\thanks{{\tt anindya@cs.berkeley.edu}.  Research supported by NSF award CCF-0915929, CCF-1017403 and CCF-1118083.}\\
University of California, Berkeley\\
\and
Ilias Diakonikolas\thanks{{\tt ilias.d@ed.ac.uk}.  This work was done while the author was at UC Berkeley supported by a Simons Postdoctoral Fellowship.}\\
University of Edinburgh\\
\and Rocco A.\ Servedio\thanks{{\tt rocco@cs.columbia.edu}. Supported by NSF grants CNS-0716245, CCF-0915929, and CCF-1115703.}\\
Columbia University\\
}

\maketitle

\begin{abstract}

For $f$ a weighted voting scheme used by $n$ voters to choose between two
candidates, the $n$ \emph{Shapley-Shubik Indices} (or {\em Shapley values})
of $f$ provide a measure of how
much control each voter can exert over the overall outcome of the vote.
Shapley-Shubik indices were introduced by Lloyd Shapley
and Martin Shubik in 1954 \cite{SS54} and are widely studied in
social choice theory as a measure of the ``influence'' of voters.
The \emph{Inverse Shapley Value Problem} is the problem of designing a weighted
voting scheme which (approximately) achieves a desired input vector of
values for the Shapley-Shubik indices.  Despite much interest in this problem
no provably correct and efficient algorithm was known prior to our work.

We give the first efficient algorithm with provable performance guarantees for 
the Inverse Shapley Value Problem.  For any constant $\eps > 0$
our algorithm runs in fixed poly$(n)$ time (the degree of the
polynomial is independent of $\eps$) and has the following
performance guarantee:  given as input a vector of desired Shapley values,
if any ``reasonable'' weighted voting scheme 
(roughly, one in which the threshold is not too skewed)\ignore{\new{
(roughly, one in which candidate has a non-trivial probability of being 
elected)}} approximately matches the desired vector of values to within 
some small error,
then our algorithm explicitly outputs a weighted voting scheme that
achieves this vector of Shapley values to within error $\eps.$
If there is a ``reasonable'' voting scheme in which all
voting weights are integers at most $\poly(n)$ that approximately achieves
the desired Shapley values, then our algorithm runs in time
$\poly(n)$ and outputs a weighted voting scheme that achieves
the target vector of Shapley values to within
error $\eps=n^{-1/8}.$

\end{abstract}

\section{Introduction}

In this paper we consider the common scenario in which each of $n$
voters must cast a binary vote for or against some proposal. 
What is the best way to design such a voting scheme? 
Throughout the paper we consider only \emph{weighted voting schemes,} in which the proposal 
passes if a weighted sum of yes-votes exceeds a predetermined threshold. 
Weighted voting schemes are predominant in voting theory and have been 
extensively studied for many years, see~\cite{Elkind:07aaai,Zuck:08aaai} and references therein.
In computer science language, we are dealing with 
\emph{linear threshold functions} (henceforth 
abbreviated as \emph{LTFs}) over $n$ Boolean variables.

If it is desired that each of the $n$ voters should have the same ``amount of power''
over the outcome, then a simple majority vote 
is the obvious solution.  However, in many scenarios it may be the case
that we would like to assign different levels of voting power
to the $n$ voters -- perhaps they are shareholders 
who own different amounts of stock in a corporation,
or representatives of differently
sized populations.  In such a setting it is much less obvious how
to design the right voting scheme; indeed, it is far from obvious
how to correctly quantify the notion of the ``amount of power''
that a voter has under a given fixed voting scheme.  As a simple
example, consider an election with three voters
who have voting weights 49, 49 and 2, in which a total of 51 votes
are required for the proposition to pass.  While the disparity between
voting weights may at first suggest that the two voters with 49 votes each
have most of the ``power,'' any coalition of two voters is
sufficient to pass the proposition and any single voter is insufficient,
so the voting power of all three voters is in fact equal.

Many different \emph{power indices} (methods of measuring
the voting power of individuals under a given %weighted
voting scheme) have been proposed over the
course of decades.  These include the Banzhaf index
\cite{Banzhaf:65}, the Deegan-Packel index \cite{DeeganPackel:78},
the Holler index \cite{Holler:82}, and others
(see the extensive survey of de Keijzer \cite{deKeijzer:08}).
Perhaps the best known,
and certainly the oldest, of these indices is the 
\emph{Shapley-Shubik index} \cite{SS54}, which is also known as the index of
\emph{Shapley values} (we shall henceforth refer to it as
such).
Informally, the Shapley value of a voter $i$ among the 
$n$ voters 
is the fraction of all $n!$ orderings of the voters
in which she ``casts the pivotal vote'' (see Definition~\ref{def:shapley}
in Section~\ref{sec:prelim}
for a precise definition, and \cite{Roth:88book} for much more
on Shapley values).  We shall work with the Shapley
values throughout this paper.

Given a particular weighted voting scheme (i.e., an $n$-variable
linear threshold function), standard sampling-based approaches
can be used to efficiently obtain highly accurate estimates of the $n$
Shapley values (see also the works of \cite{Leech:03b,BMRPRS10}).  
However, the \emph{inverse} problem 
is much more challenging:  given a vector of $n$ desired
values for the Shapley values, how can one design a weighted
voting scheme that (approximately) achieves these Shapley values?
This problem, which we refer to
as the \emph{Inverse Shapley Value Problem}, is quite natural and
has received considerable attention; various heuristics
and exponential-time algorithms 
have been proposed \cite{APL:07,FWJ08short,KKZ10,Kurz11}, but prior to our 
work no provably correct and efficient algorithms were known.

\medskip

\noindent {\bf Our Results.}
We give the first efficient
algorithm with provable performance guarantees
for the Inverse Shapley Value Problem. 
Our results apply to ``reasonable'' voting schemes;
roughly, we say that a weighted voting scheme is ``reasonable'' if
fixing a tiny fraction of the voting weight does not already determine 
the outcome, 
i.e., if the threshold of the linear threshold function is not too
extreme.  (See Definition~\ref{def:reasonable} in Section~\ref{sec:prelim}
for a precise definition.) 
This seems to be a plausible property for natural
voting schemes.   Roughly speaking, we show that if there
is any reasonable weighted voting scheme
\ignore{\new{We call a weighted voting scheme 
(monotone linear threshold function) {\em reasonable} 
if each candidate has probability $2^{-o(n)}$ to be elected -- under the ``impartial culture
assumption''~\cite{GK:68} -- where $n$ is the number of voters. (See Section...for a formal definition.) 
We feel that this is a common sense property of 
voting schemes. For all such voting schemes,}}that approximately
achieves the desired
input vector of Shapley values, then our algorithm finds such a 
weighted voting scheme.  Our algorithm runs in fixed
polynomial time in $n$, the number of voters, for any
constant error parameter $\eps > 0$.  
In a bit more detail, our first main theorem, stated informally, is as follows
(see Section~\ref{sec:mainresults} for Theorem~\ref{thm:main-arbitrary}
which gives a precise theorem statement):

\medskip

\noindent {\bf Main Theorem (arbitrary weights, informal statement).}
\emph{
There is a poly$(n)$-time algorithm
with the following properties:  The algorithm is given any constant
accuracy parameter $\eps > 0$ and any vector of
$n$ real values
$a(1),\dots, a(n)$.  The algorithm
has the following performance guarantee:
if there is any monotone increasing reasonable LTF $f(x)$
whose Shapley values are very close to the given values
$a(1),\dots, a(n)$,
then with very high probability the algorithm 
outputs $v \in \R^n,$ $\theta \in \R$ such
that the linear threshold function 
$h(x)=\sign(v \cdot x - \theta)$ has Shapley values $\eps$-close
to those of $f$.
}

\smallskip

\new{We emphasize that the exponent of the $\poly(n)$ running time is a fixed constant that is independent of $\eps$.}

\medskip

Our second main theorem gives an even stronger
guarantee if there is a weighted voting scheme with small weights
(at most $\poly(n)$) whose Shapley values are close to the desired values.
For this problem we give an algorithm which achieves $1/\poly(n)$
accuracy in $\poly(n)$ time.
An informal statement of this result
is (see Section~\ref{sec:mainresults} for Theorem~\ref{thm:main-bounded} 
which gives a precise theorem statement):

\medskip

\noindent {\bf Main Theorem (bounded weights, informal statement).}
\emph{
There is a poly$(n,W)$-time algorithm
with the following properties:  The algorithm is given a weight
bound $W$ \ignore{, an accuracy parameter $\eps = \Omega(n^{-1/3})$}and 
any vector of $n$ real values
$a(1),\dots, a(n)$.  The algorithm
has the following performance guarantee:
if there is any monotone increasing reasonable LTF $f(x)=
\sign(w \cdot x - \theta)$
whose Shapley values are very close to the given values
$a(1),\dots, a(n)$
and where each $w_i$ is an integer of magnitude at most $W$,
then with very high probability the algorithm 
outputs $v \in \R^n,$ $\theta \in \R$ such
that the linear threshold function 
$h(x)=\sign(v \cdot x - \theta)$ has Shapley values $n^{-1/8}$-close
to those of $f$.
}
\ignore{
\rnote{should we mention here in the intro that for both
our algorithms, except
with negligible failure probability, they will never output a weighted
voting scheme whose Shapley vector is far from the desired solution?
My sense is that we don't need to; if you agree, someone erase this note.}
}

\medskip

\noindent {\bf Discussion and Our Approach.}
At a high level, the Inverse Shapley Value Problem that we consider is 
similar to the ``Chow Parameters Problem'' that has been the subject 
of several recent papers \cite{Goldberg:06bshort,OS08short,DDFS12}.  
The Chow parameters are another name for the $n$ Banzhaf indices; 
the Chow Parameters Problem is to output a 
linear threshold function which approximately matches a given input 
vector of Chow parameters.   (To align with the terminology of the current
paper, the ``Chow Parameters Problem'' might perhaps better be described as
the ``Inverse Banzhaf Problem.'')

Let us briefly describe the approaches in \cite{OS08short} and
\cite{DDFS12} at a high level for the purpose of establishing
a clear comparison with this paper.  
Each of the papers \cite{OS08short,DDFS12} combines structural
results on linear threshold functions with an algorithmic
component.  The structural results in \cite{OS08short} deal with
anti-concentration of affine forms $w \cdot x - \theta$ where 
$x \in \{-1,1\}^n$ is uniformly distributed over the Boolean
hypercube, while the algorithmic ingredient of \cite{OS08short}
is a rather straightforward brute-force search.
In contrast, the key structural results of \cite{DDFS12} are geometric
statements about how $n$-dimensional hyperplanes interact with the
Boolean hypercube, which are combined with linear-algebraic
(rather than anti-concentration) arguments.
The algorithmic ingredient of \cite{DDFS12} is more
sophisticated, employing a boosting-based approach inspired
by the work of \cite{TTV:09short,Impagliazzo:95short}.

Our approach combines aspects of both the \cite{OS08short} and
\cite{DDFS12} approaches.  Very roughly speaking, we establish
new structural results which show that linear threshold functions have
good anti-concentration (similar to \cite{OS08short}), and use a boosting-based
approach derived from \cite{TTV:09short} as the algorithmic component
(similar to \cite{DDFS12}).  However, 
this high-level description glosses over
many ``Shapley-specific'' issues and complications that do not arise in these earlier
works; below we describe two of the main challenges that arise,
and sketch how we meet them in this paper.

\medskip

\noindent {\bf First challenge:  establishing anti-concentration with respect
to non-standard distributions.}
The Chow parameters (i.e., Banzhaf indices) 
have a natural definition in terms of
the uniform distribution over the Boolean hypercube $\{-1,1\}^n$.
Being able to use the uniform distribution with its many nice properties
(such as complete independence among all coordinates) 
is very useful in proving the required anti-concentration results
that are at the heart of \cite{OS08short}.  In contrast, it is not
\emph{a priori} clear what is (or even whether there exists) 
the ``right'' distribution over
$\{-1,1\}^n$ corresponding to the Shapley values.  In this paper
we derive such a distribution $\mu$ over $\{-1,1\}^n$,
 but it is much less well-behaved than the uniform distribution 
(it is supported on a proper subset of $\{-1,1\}^n$, and it is not even
pairwise independent).  Nevertheless, we are able to establish
anti-concentration results for affine forms $w \cdot x - \theta$
corresponding to linear threshold functions under the
distribution $\mu$ as
required for our results.  This is done by showing that any
\newer{reasonable}
linear threshold function can be expressed with ``nice''
weights (see Theorem~\ref{thm:LParg} of Section~\ref{sec:prelim}),
and establishing anti-concentration for any ``nice'' weight vector
by carefully combining anti-concentration bounds for $p$-biased distributions
across a continuous family of different choices of $p$
(see Section~\ref{sec:str} for details).

\medskip

\noindent {\bf Second challenge:  using anti-concentration to solve
the Inverse Shapley problem.} The 
main algorithmic ingredient that we use is a procedure 
from \cite{TTV:09short}.  Given a vector of values 
$(\E[f(x)x_i])_{i=1,\dots,n}$
(correlations between the unknown linear threshold function 
$f$ and the individual input
variables), it efficiently constructs a bounded function $g: \{-1,1\}^n
\to [-1,1]$ which closely matches these correlations, i.e.,
$\E[f(x)x_i] \approx \E[g(x)x_i]$ for all $i$.  Such a procedure is very
useful for the Chow parameters problem, because the
Chow parameters correspond precisely to the values $\E[f(x)x_i]$ --
i.e., the degree-$1$ Fourier coefficients of $f$ --
with respect to the uniform distribution.
(This correspondence is at the heart
of Chow's original proof \cite{Chow:61short} showing that the 
{exact} values of the Chow parameters suffice to information-theoretically
specify any linear threshold function;
anti-concentration is used in \cite{OS08short} to extend Chow's original 
arguments about degree-1 Fourier coefficients to the setting of approximate 
reconstruction.)

For the inverse Shapley problem, there is no obvious correspondence
between the correlations of individual input variables and the
Shapley values.  Moreover, without a notion of ``degree-$1$ Fourier 
coefficients'' for the Shapley setting, it is not clear why 
anti-concentration statements with respect to $\mu$ should be useful for 
approximate reconstruction.  We deal with both these issues by 
developing a notion of the \emph{degree-$1$ Fourier coefficients of $f$
with respect to distribution $\mu$} and relating these coefficients
to the Shapley values
\footnote{We note that Owen \cite{Owen:72} has given a characterization of
the Shapley values as a weighted average of $p$-biased influences
(see also \cite{KalaiSafra:06}).  However, this is not  
as useful for us as our characterization in terms of 
``$\mu$-distribution'' Fourier coefficients, because we need
to ultimately relate the Shapley values to anti-concentration
with respect to $\mu$.}.
(We actually require two related notions:  one is
the ``coordinate correlation coefficient'' $\E_{x \sim \mu}[f(x)x_i]$, 
which is necessary for the algorithmic \cite{TTV:09short} ingredient,
and one is the ``Fourier coefficient'' $\hat{f}(i)= \E_{x \sim \mu}
[f(x) L_i]$, which is necessary for Lemma~\ref{lem:anticonc-and-ell1}, see below.)
We define both notions and establish the necessary relations between
them in Section~\ref{sec:reformulation}. 

Armed with the notion of the degree-$1$ Fourier coefficients
under distribution $\mu$, we prove a key result
(Lemma~\ref{lem:anticonc-and-ell1}) saying that if the LTF 
$f$ is anti-concentrated under distribution $\mu$, 
then any bounded function $g$ which closely matches
the degree-$1$ Fourier coefficients of $f$ must be close to
$f$ in $\ell_1$ distance with respect to $\mu$.  
(This is why anti-concentration with respect to $\mu$ is useful for us.)
From this point, exploiting properties of the
\cite{TTV:09short} algorithm, we can pass from $g$
to an LTF whose Shapley values closely match those of $f$.

\medskip

\noindent {\bf Organization.}  Useful preliminaries are given in Section~\ref{sec:prelim},
including the crucial fact
(Theorem~\ref{thm:LParg}) that all ``reasonable''
linear threshold functions have weight representations with ``nice'' weights.
In Section~\ref{sec:reformulation} we 
define the distribution $\mu$ and the notions of Fourier
coefficients and ``coordinate correlation coefficients,''
and the relations between them,
that we will need.  At the end of that section we prove
a crucial lemma, Lemma~\ref{lem:anticonc-and-ell1}, which says 
that anti-concentration of affine forms and
closeness in Fourier coefficients together suffice to establish
closeness in $\ell_1$ distance.  Section~\ref{sec:str} proves that
``nice'' affine forms have the required anti-concentration, 
and Section~\ref{sec:algorithmic} describes the algorithmic
tool from \cite{TTV:09short} that lets us establish closeness of 
coordinate correlation coefficients.  Section~\ref{sec:mainresults}
puts the pieces together to prove our main theorems. Finally, 
in Section~\ref{sec:concl}
we conclude the paper and present a few open problems.

\section{Preliminaries} \label{sec:prelim}

\noindent {\bf Notation and terminology.}
For $n \in \Z_{+}$, we denote by $[n] \eqdef \{1, 2, \ldots, n\}$.
For $i, j \in \Z_{+}$, $i \le j$, we denote $[i, j] \eqdef \{i, i+1, \ldots, j\}.$

Given a vector $w = (w_1,\dots,w_n) \in \R^n$ we write $\|w\|_1$
to denote $\sum_{i=1}^n |w_i|$.  A \emph{linear threshold function},
or LTF, is a function
$f: \{-1,1\}^n \to \{-1,1\}$ which is such that
$f(x) = \sign(w \cdot x - \theta)$ for some
$w \in \R^n, \theta \in \R.$

Our arguments will also use a variant of linear threshold functions which
we call \emph{linear bounded functions} (LBFs).  The projection
function $P_1: \R \to [-1,1]$ is defined by $P_1(t)=t$ for $|t| \leq 1$
and $P_1(t)=\sign(t)$ otherwise.  An LBF $g: \{-1,1\}^n \to [-1,1]$ is a
function $g(x)=P_1(w \cdot x - \theta).$

\medskip

\noindent {\bf Shapley values.}
Here and throughout the paper we write $\ms$ to denote the symmetric
group of all $n!$ permutations over $[n].$ Given a
permutation $\pi \in \ms$ and an index $i \in [n]$, we write $x(\pi,i)$
to denote the string in $\{-1,1\}^n$ that has a 1 in coordinate $j$ if
and only if $\pi(j) < \pi(i)$, and we write $x^{+}(\pi,i)$ to denote the
string obtained from $x(\pi,i)$ by flipping coordinate $i$ from $-1$ to
$1.$ With this notation in place we can define the generalized Shapley
indices of a Boolean function as follows:
\begin{definition} {\bf (Generalized Shapley values)}
\label{def:shapley}
Given $f : \{-1,1\}^n \to \{-1,1\},$ the \emph{$i$-th
generalized Shapley value of $f$} is the value
\begin{equation}
\tilde{f}(i) \eqdef \E_{\pi \sim_R \ms}[f(x^+(\pi,i)) - f(x(\pi,i))]
\label{eq:shapley}
\end{equation}
\new{(where ``$\pi \sim_R \ms$'' means that $\pi$ is selected uniformly at random from $\ms$).}
\end{definition}
A function $f : \{-1,1\}^n \rightarrow \{-1,1\}$ is said to be
\emph{monotone increasing} if for all $i \in [n]$, whenever 
two input strings $x, y \in \{-1,1\}^n$
differ precisely in coordinate $i$ and have $x_i=-1$, $y_i=1$, it is the
case that $f(x) \leq f(y).$ It is easy to check that for monotone functions
our definition of generalized Shapley values agrees with the usual notion of 
Shapley values (which are typically defined only for monotone functions) 
up to a multiplicative factor of 2; in
the rest of the paper we omit ``generalized'' and refer to these values
simply as the Shapley values of $f.$

We will use the following notion of the ``distance'' between the vectors of Shapley values for two
functions $f,g: \{-1,1\}^n \to [-1,1]$:
\[
\dshap(f,g) \eqdef \sqrt{\littlesum_{i=1}^n (\tilde{f}(i) - \tilde{g}(i))^2},\]
i.e., the Shapley distance $\dshap(f,g)$ is simply the Euclidean
distance between the two $n$-dimensional vectors of Shapley values.  Given
a vector $a = (a(1),\dots, a(n)) \in \R^n$ we will also
use $\dshap(a,f)$ to denote $\sqrt{\sum_{i=1}^n (\tilde{f}(i)-a(i))^2}.$

\medskip

\noindent {\bf The linear threshold functions that we consider.}
Our algorithmic results hold for linear threshold functions which are
not too ``extreme'' (in the sense of 
\newer{
having a very skewed threshold}).
\ignore{almost always outputting +1 or
almost always outputing $-1$).} We will use the following definition:

\begin{definition}  \label{def:reasonable}
{\bf ($\eta$-reasonable LTF)}
Let $f: \{-1,1\}^n \to \{-1,1\}$, $f(x)=\sign(w \cdot x - \theta)$
be an LTF.  For $0 < \eta < 1$ we say that $f$ is \emph{$\eta$-reasonable}
if $\theta \in [-(1-\eta) \|w\|_1, (1-\eta) \|w\|_1].$
\end{definition}

All our results will deal with $\eta$-reasonable LTFs; throughout
the paper $\eta$ should be thought of as a small fixed absolute constant 
(such as $1/1000$).
LTFs that are not $\eta$-reasonable do not seem to correspond to
very interesting voting schemes since typically they will be very
close to constant functions.  (For example, even at $\eta = 0.99$, if
the LTF $f(x)=\sign(x_1 + \cdots + x_n - \theta)$ has a threshold 
\newer{$\theta>0$ which makes it not an $\eta$-reasonable LTF, then $f$ agrees
with the constant function $-1$ on all but a $2^{-\Omega(n)}$ fraction of 
inputs in $\{-1,1\}^n.$})
\ignore{
$\theta$ which makes it not a $\eta$-reasonable LTF, then
then $f$ is $2^{-\Omega(n)}$-close to either the constant function $+1$
or $-1$.)
}

Turning from the threshold to the weights,
some of the proofs in our paper will require us to work with LTFs that have
``nice'' weights in a certain technical sense.
Prior work~\cite{Servedio:07cc,OS11:chow} has shown that for any LTF, there is
a weight vector realizing that LTF that has essentially the properties
we need; however, since the exact technical condition that we require is not
guaranteed by any of the previous works, we give a full proof that any LTF
has a representation of the desired form.
The following theorem is proved in Appendix~\ref{app:niceweights}:

\begin{theorem}\label{thm:LParg} 
Let $f: \bn \to \bits$ be an $\eta$-reasonable LTF and $k \in [2, n]$.
There exists a representation of $f$ as $f(x) = \sgn(v_0+ \littlesum_{i=1}^n 
v_i x_i)$ such that (after reordering \new{coordinates} so that condition (i) below  holds) we have:   (i) $|v_{i}| \geq |v_{i+1}|$, $i \in [n-1]$; (ii) $|v_0| \leq (1-\eta) \littlesum_{i=1}^n |v_i|$; and (iii) for all $i \in [0, k-1]$ we have $|v_{i}| \leq (2/\eta) \cdot \sqrt{n}
 \cdot k^{\frac{k}{2}} \cdot \sigma_k$,
where $\sigma_k \eqdef \sqrt{\littlesum_{j \geq k} v_j^2}$.
\end{theorem}

\noindent {\bf Tools from probability.}
We will use the following standard tail bound:

\begin{theorem}\label{thm:chernoff} {\bf (Chernoff Bounds)} Let $X$ be
a random variable taking values in $[-a,a]$ and let $X_1, \ldots, X_t$
be i.i.d. samples drawn from $X$. Let $\overline{X} = \sum_{i=1}^t X_i / t$. Then for any $\gamma>0$, we have
$$ \Pr \left[\left|\overline{X} -
\mathbf{E}[X] \right| \ge \gamma\right] \le2 \exp(-\gamma^2 t/(\newer{2}a^2)). 
$$
\end{theorem}

We will also use the Littlewood-Offord inequality for
$p$-biased distributions over $\{-1,1\}^n.$  One way to prove this is by using the
LYM inequality (which can be found e.g. as Theorem~8.6 of
\cite{Jukna01book}); for an explicit reference and proof of the following statement see
e.g.~\cite{AGKW:09}.

\begin{theorem} Fix $\delta \in (0,1)$ and let
$D_{\delta}$ denote the {\em $\delta$-biased distribution over $\{-1,1\}^n$}
(under which each coordinate is set to $1$ independently
with probability $\delta.$)
Fix $w \in \mathbb{R}^n$ and define $ S = \{i : |w_i | \ge
\epsilon \} $. If $|S| \ge K$, then for all $\theta \in \mathbb{R}$ we have
$\Pr_{x \sim D_{\delta}} [|w \cdot x - \theta| < \epsilon] \le
\frac{1}{\sqrt{K \delta(1-\delta)}}.$ \end{theorem} 

\medskip

\noindent
{\bf Basic Facts about function spaces.}
We will use the following basic facts: 

\begin{fact}\label{fac:func1}\ignore{
\rnote{This had been

``
Let $\mu$ be a probability distribution over $\{-1,1\}^n\setminus \{\mathbf{1}, \mathbf{-1} \}$ such that $\forall x \in \{-1,1\}^n$, $\mu(x)>0$. Also, define $f_0 : \{-1,1\}^n\setminus \{\mathbf{1}, \mathbf{-1} \} \rightarrow \mathbb{R}$ as $f_0 : x \mapsto 1$. Similarly, for $1 \le i \le n$, define $f_i : \{-1,1\}^n\setminus \{\mathbf{1}, \mathbf{-1} \} \rightarrow \mathbb{R}$ as $f_i : x \mapsto x_i$. Then the set $\{f_0, f_1 ,\ldots, f_n \}$ is a set of linearly independent functions and forms a basis for the subspace $V = \{f : \{-1,1\}^n \rightarrow \mathbb{R} \textrm{ and $f$ is linear } \}$.
''

but the above seemed simpler and equivalent for our purposes, is this correct?
I think the previous version was a little off, it didn't use $\mu$ at all
after mentioning it initially.}}
The $n+1$ functions $1,x_1,\dots,x_n$
are linearly independent and form a basis for the subspace
$V = \{f : \{-1,1\}^n \rightarrow \mathbb{R} \textrm{ and $f$ is linear
} \}$.
\end{fact}

\begin{fact}\label{fac:func2} Fix any $\Omega \subseteq \{-1,1\}^n$ and
let $\mu$ be a probability distribution over $\Omega$
such that $\mu(x)>0$ for all $x \in \Omega.$
We define $\braket{f,g}_\mu \eqdef \mathbf{E}_{\omega \sim \mu} [ f(\omega)
g(\omega)]$ for $f, g : \Omega \rightarrow \mathbb{R}$.
Suppose that $f_1, \ldots, f_m : \Omega \rightarrow
\mathbb{R}$ is an orthonormal set of functions, i.e.,
$\braket{f_i,f_j}_{\mu}=\delta_{ij}$ for all $i,j \in [m].$  Then we have
$
\braket{f,f}_{\mu}^2 \ge \sum_{i=1}^m \braket{f,f_i}_{\mu}^2. $
As a corollary, if $f,h: \Omega \to \{-1,1\}$ then we have
$\sqrt{\sum_{i=1}^m \langle f-h, f_i \rangle^2_\mu}
\leq 2 \sqrt{\Pr_{x \sim \mu}[f(x) \neq h(x)]}.$
\end{fact}

% %%%%%%%%%%%%%%%%%%%%%%%%%%%%%%%%%%%%%%%%%%%%%%%%%%%%%%%%%%%%%%%%%%%%%%%%%
% End of prelim
% %%%%%%%%%%%%%%%%%%%%%%%%%%%%%%%%%%%%%%%%%%%%%%%%%%%%%%%%%%%%%%%%%%%%%%%%%

\section{\new{Analytic} Reformulation of Shapley values}
\label{sec:reformulation}

The definition of Shapley values given in Definition~\ref{def:shapley}
is somewhat cumbersome to work with.
In this section we derive alternate characterizations of
Shapley values in terms of ``Fourier coefficients''
and ``coordinate correlation coefficients''
 and establish various technical results relating
Shapley values and these coefficients;
these technical results will be crucially used 
in the proof of our main theorems.

There is a particular distribution $\mu$ that plays a central
role in our reformulations.  We start by defining this 
distribution $\mu$ and introducing some relevant notation, and then
give our results.  
%Because of space constraints all proofs are in Appendix~\ref{ap:technical}.

\medskip

\noindent {\bf The distribution $\mu$.}
Let us define $\Lambda(n) \eqdef \sum_{0<k<n} \frac 1k + \frac1{n-k}$; clearly we have $\Lambda(n) = \Theta(\log n)$,
and more precisely we have $\Lambda(n) \leq 2 \log n.$
We also define $Q(n,k)$ as $Q(n,k) \eqdef \frac 1k + \frac1{n-k}$ for  $0<k<n$,
so we have $\Lambda(n) = \sum_{k=1}^{n-1} Q(n,k).$

For $x \in \{-1,1\}^n$ we write $\weight(x)$ to denote the number of $1$'s in $x$.
We define the set $B_n$ to be $B_n \eqdef \{x \in \{-1,1\}^n: 0 < \weight(x) < n\}$, i.e.,
$B_n=\{-1,1\}^n \setminus \{\mathbf{1}, \mathbf{-1}\}$.

The distribution $\mu$ is supported on $B_n$ and is defined as follows: to make a draw from $\mu$,
sample $k \in \{1,\ldots, n-1\}$ with probability $Q(n,k)/\Lambda(n)$. Choose $x \in \{-1,1\}^n$ uniformly at random from the $k$-{th} 
``weight level'' of  $\{-1,1\}^n$, i.e.,
from $\{-1,1\}^n_{=k} \eqdef \{x \in \{-1,1\}^n: \weight(x)=k\}.$

\medskip

\noindent {\bf Useful notation.}  For $i=0,\dots,n$ we define 
the ``coordinate correlation coefficients'' of a function $f:
\{-1,1\}^n \to \R$ (with
respect to $\mu$) as:
\begin{equation} \label{eq:ccc}
f^*(i) \eqdef \E_{x \sim \mu}[f(x) \cdot x_i]
\end{equation}
(here and throughout the paper $x_0$ denotes the constant 1).

Later in this section we will define an orthonormal set of linear functions
$L_0,L_1,\dots,L_n: \{-1,1\}^n \to \R$.  We define
the ``Fourier coefficients'' of $f$ (with respect to $\mu$) as:
\begin{equation} \label{eq:fouriercoeffs}
\hat{f}(i) \eqdef \E_{x \sim \mu}[f(x)  \cdot L_i(x)].
\end{equation}

\ignore{
We begin with some intuition before entering into the detailed
result of this section.
We first give expressions for the Shapley values of $f$ in terms
of the correlation coefficients $f^*(i)$ (Proposition~\ref{prop:expshap}).
We then define the desired set of functions $L_0 = 1, L_1,\ldots, L_n$
and show that they comprise an orthonormal set of functions under $\mu$
(Lemma~\ref{lem:ortho}).  
Then we show that
for any  function $f : \{-1,1\}^n \rightarrow \{-1,1\}$,
the Shapley value $\tilde{f}(i)$ will closely correspond to
$\hat{f}(i) \eqdef \braket{f,L_i}_\mu$, i.e., the ``$i$-th degree-1 Fourier coefficient of $f$
with respect to $\mu$'' (Lemma~\ref{lem:foushap}).  Next,
we define two different
distance notions between functions $f$ and $g$, the ``Fourier distance''
and the ``Shapley distance,''  and use the preceding results to upper bound
the Shapley distance in terms of the Fourier distance
(Lemma~\ref{lem:shapclose}).
Finally, we show that if the Fourier distance between an LTF $f$ and and
a bounded function $g$ is small and the affine form defining $f$ has
``good anti-concentration,'' then the $\ell_1$ distance (with
respect to distribution
$\mu$) must also be small.\rnote{Here or elsewhere, fit this into the
larger picture of our overall proof structure.}
}

\medskip

%\subsection{Technical results on Shapley
%indices, Fourier coefficients, and coordinate correlation coefficients}

\noindent {\bf An alternative expression for the
Shapley values.}
We start by expressing the Shapley values in terms
of the coordinate correlation coefficients:

\begin{lemma}\label{lem:expshap}
Given $f : \{-1,1\}^n \rightarrow [-1,1]$, for each 
$i=1,\dots,n$ we have
$$
\tilde{f}(i) = \frac{f(\mathbf{1} ) - 
f(\mathbf{-1})}{n} + \frac{\Lambda(n)}{2} \cdot \left( f^*(i)  - {\frac 1 n}
\littlesum_{j=1}^n f^*(j)\right),
$$
or equivalently,
$$
f^*(i) =
{\frac 2{\Lambda(n)}} \cdot \left(\tilde{f}(i) - \frac{f(\mathbf{1} ) - 
f(\mathbf{-1})}{n}\right) + {\frac 1 n}
\littlesum_{j=1}^n f^*(j).
$$
\end{lemma}

\begin{proof}
Recall that $\tilde{f}(i)$ can be expressed as follows:
\begin{equation}
\tilde{f}(i) = \E_{\pi \sim_R \ms}[f(x^+(\pi,i)) - f(x(\pi,i))]. \label{eq:shap}
\end{equation}
Since the $i$-th coordinate of $x^+(\pi,i)$ is $1$ and the $i$-th coordinate of
$x(\pi,i)$ is $-1$, we see that $\tilde{f}(i)$ is a weighted sum of $\{f(x) x_i \}_{x \in \{-1,1\}^n}$. We now compute the weights associated with any such $x \in \{-1,1\}^n$.

\begin{itemize}
\item Let $x$ be a string that has $\weight(x)$ coordinates that are 1 and has $x_i=1.$  Then the total number of permutations $\pi \in \ms$ such that $\xip = x$ is $(\weight(x)-1)! (n-\weight(x))!$. Consequently the weight associated with $f(x) x_i$ for such an $x$ is $(\weight(x)-1)!
    \cdot  (n-\weight(x))! /n!$.
\item Now let $x$ be a string that has $\weight(x)$ coordinates that are 1 and has $x_i=-1.$  Then the total number of permutations $\pi \in \ms$ such that $\xop = x$ is $\weight(x)! (n-\weight(x)-1)!$. Consequently the weight associated with $f(x) x_i$ for such an $x$ is $\weight(x)! \cdot (n-\weight(x)-1)!/n!$.
\end{itemize}
Thus we may rewrite Equation~(\ref{eq:shap}) as
\begin{eqnarray*}
\tilde{f}(i) &=& \sum_{x: \{-1,1\}^n : x_i =1} \frac{(\weight(x)-1)! (n-\weight(x))!}{n!} f(x) \cdot x_i +\\
& &  \sum_{x: \{-1,1\}^n : x_i =-1} \frac{\weight(x)! (n-\weight(x)-1)!}{n!} f(x) \cdot x_i.
\end{eqnarray*}
Let us now define $\nu(f) \eqdef (f(\mathbf{1}) - f(\mathbf{-1}))/n$.   Using the fact that $x_i^2=1$, it is easy to see that one gets
\begin{eqnarray}
2\tilde{f}(i) &=& 2 \nu(f) + \nonumber \\
& &2 \left( \sum_{x \in B_n} f(x) \cdot \frac{(\weight(x)-1)! (n-\weight(x)-1)!}{n!} \cdot ((n/2-\weight(x))  + (nx_i)/2) \right)  \nonumber \\
&=& 2 \nu(f) + \sum_{x \in B_n} \left(f(x) \cdot \frac{(\weight(x)-1)! (n-\weight(x)-1)!}{(n-1)!} \cdot x_i + \right.
\nonumber \\
& & \left.\text{~~~~~~~~~~~~~~~~~~~~~~~~~}f(x) \cdot \frac{(\weight(x)-1)! (n-\weight(x)-1)!}{n!} \cdot (n-2\weight(x))\right) \nonumber \\
&=& 2 \nu(f) + \sum_{x \in B_n} \left(f(x) \cdot \frac{n}{\weight(x)(n-\weight(x))\binom{n}{\weight(x)}} \cdot x_i + \right. \nonumber \\
& &\left.\text{~~~~~~~~~~~~~~~~~~~~~~~~~} f(x) \cdot \frac{1}{\weight(x)(n-\weight(x)) \binom{n}{\weight(x)}} \cdot (n-2\weight(x))\right). \label{eq:shap2}
\end{eqnarray}
We next observe that $n-2\weight(x) =-( \littlesum_{j \in [n]} x_j)$.
Next, let us define \new{$P(n, k)$ (for $k \in [1, n-1]$)} as follows :
\new{$$
P(n, k ) \eqdef \frac{Q(n, k)}{\binom{n}{k}} = 
\frac{\frac{1}{k} + \frac{1}{n-k}}{\binom{n}{k}}. $$}
So we may rewrite Equation~(\ref{eq:shap2}) in terms of $P(n,\weight(x))$ as
$$
2 \tilde{f}(i) =2 \nu(f) +  \sum_{x \in B_n} \left[ f(x) \cdot x_i \cdot P(n,\weight(x)) \right] - \sum_{x \in B_n} \left[ f(x) \cdot P(n,\weight(x)) \cdot (\littlesum_{i=1}^n x_i)/n \right].
$$
We have $$\sum_{x \in B_n} P(n, \weight(x)) = \new{ \sum_{k=1}^{n-1} \sum_{x \in \bn_{=k}} P(n, \weight(x)) = \sum_{k=1}^{n-1} \binom{n}{k} \cdot P(n, k)=  \sum_{k=1}^{n-1} Q(n, k) }  = \Lambda(n),$$ and consequently we get
$$
2 \tilde{f}(i) = 2 \nu(f) + \Lambda(n) \cdot \left( \mathop{\mathbf{E}}_{x \sim \mu} \left[ f(x) \cdot x_i \right] - \mathop{\mathbf{E}}_{x \sim \mu} \left[ f(x) \cdot (\littlesum_{i=1}^n x_i)/n\right] \right),
$$
finishing the proof.
\end{proof}

\medskip

\noindent
{\bf Construction of a Fourier basis for distribution $\mu$.}
For all $x \in B_n$ we have that $\mu(x)>0$, and consequently by 
Fact~\ref{fac:func1}
we know that the functions $1, x_1,\dots,x_{n+1}$ 
form a basis for the subspace of linear functions from $B_n \rightarrow 
\mathbb{R}$. By Gram-Schmidt orthogonalization, we can obtain an orthonormal 
basis $L_0,\dots,L_n$ for this subspace, i.e.,  a set of linear functions such that
$\braket{L_i,L_i}_\mu =1 \text{~for all~}i$ and $\braket{L_i,L_j}_{\mu} = 0
$ for all $i \neq j.$

We now give explicit expressions for these basis functions.
We start by defining $L_0 :B_n \rightarrow \mathbb{R}$ as $L_0 : x \mapsto 1$. 
Next, by symmetry, we can express each $L_i$ as
$$
L_i(x) =\alpha(x_1 + \ldots +x_n) + \beta x_i.
$$
Using the orthonormality properties it is straightforward to solve for $\alpha$ 
and $\beta$.  The following Lemma
gives the values of $\alpha$ and $\beta$:
\begin{lemma}\label{lem:ortho}
For the choices
$$
\alpha \eqdef \frac{1}{n} \cdot \left( \sqrt{\frac{\Lambda(n)}{n \Lambda(n)  - 4 (n-1)}}  -\frac{\sqrt{\Lambda(n)}}{2} \right), \quad  \beta \eqdef \frac{\sqrt{\Lambda(n)}}{2},$$
the set $\{L_i \}_{i=0}^n$ is an orthonormal set of linear functions under the distribution $\mu$.
\end{lemma}

\noindent We note for later reference that $\alpha = - \Theta\left({\frac {\sqrt{\log n}} n}\right)$
and $\beta = \Theta(\sqrt{\log n}).$

We start with the following proposition which gives an explicit expression for $\mathop{\mathbf{E}}_{x \sim \mu} [x_ix_j]$ when $i \not =j$;
we will use it in the proof of Lemma~\ref{lem:ortho}.

\begin{proposition}\label{prop:corr}
%Let $\mu$ be the probability distribution over $\{-1,1\}^n$ that is defined in Section~\ref{sec:reformulation}. Then 
For all  $1\le i < j \le n$ we have 
$\E_{x \sim \mu} [x_ix_j]  = 1 - \frac{4}{\Lambda(n)}$.
\end{proposition}

\begin{proof}
For brevity let us write $A_k = \{-1,1\}^n_{=k}$, i.e.,  $A_k = \{x \in \{-1,1\}^n : \weight(x)=k\}$, the
 $k$-th ``slice'' of the hypercube.  Since $\mu$ is supported on $B_n = \cup_{k=1}^{n-1} A_k$, we have 
\[
\E_{x \sim \mu} [x_ix_j ] = \sum_{0<k<n} \mathop{\mathbf{E}}_{x \sim \mu} [x_ix_j\ | \ x \in A_k] \cdot \Pr_{x \sim \mu} [x \in A_k].
\]
If  $k=1$ or $n-1$, it is clear that
$$
\E_{x \sim \mu} [x_ix_j \ | \ x \in A_k] = 1- \frac{2}{n} - \frac{2}{n} = 1- \frac{4}{n},
$$
and when $2 \le k \le n-2$, we have
$$
\E_{x \sim \mu} [x_ix_j \ | \ x \in A_k]  = \frac{1}{\binom{n}{k}} \cdot \left( 2 \binom{n-2}{k-2} + 2 \binom{n-2}{k}  - \binom{n}{k} \right).
$$
Recall that $\Lambda(n) =  \sum_{0<k<n} \frac{1}{k} + \frac{1}{n-k}$ and $Q(n,k) =\frac{1}{k} + \frac{1}{n-k}$ for $0 < k < n.$  This means that
we have $$\Pr_{x \sim \mu}[x \in A_k] = Q(n,k)/\Lambda(n).$$
Thus we may write $\E_{x \sim \mu}[x_i x_j]$ as
\begin{eqnarray*}
\E_{x \sim \mu} [x_ix_j] &=& \sum_{2 \le k \le n-2} \frac{Q(n,k)}{\Lambda(n)}  \cdot \E_{x \sim \mu} [x_ix_j \ | \ x \in A_k]  + \\
& &\sum_{k \in \{1,n-1\}} \frac{Q(n,k)}{\Lambda(n)}  \cdot 
\E_{x \sim \mu} [x_ix_j \ | \ x \in A_k].
\end{eqnarray*}
For the latter sum, we have
$$
\sum_{k \in \{1,n-1\}} \frac{Q(n,k)}{\Lambda(n)}  \cdot \E_{x \sim \mu} [x_ix_j \ | \ x \in A_k] =\frac{1}{ \Lambda(n)}   \left(1- \frac{4}{n} \right) \cdot \frac{2n}{n-1}.
$$
For the former, we can write
\begin{eqnarray*}
&& \sum_{k=2}^{n-2} \frac{Q(n,k)}{\Lambda(n)}  \cdot \E_{x \sim \mu} [x_ix_j \ | \ x \in A_k]\\ 
&=& \sum_{k=2}^{n-2}\frac{1}{\Lambda(n)} \frac{(k-1)! (n-k-1)!}{(n-1)!}\cdot \left( 2 \binom{n-2}{k-2} + 2 \binom{n-2}{k}  - \binom{n}{k} \right)  \\
&=& \sum_{k=2}^{n-2} \frac{1}{\Lambda(n)} \cdot \left( \frac{2(k-1)}{(n-1)(n-k)}  +\frac{2(n-k-1)}{(n-1)k} -\frac{n}{k(n-k)}\right) \\
&=& \sum_{k=2}^{n-2}\frac{1}{\Lambda(n)} \cdot \left( \frac{2}{n-k} - \frac{2}{n-1} + \frac{2}{k} - \frac{2}{n-1} - \frac{1}{k} - \frac{1}{n-k} \right) \\
&=& \sum_{k=2}^{n-2}\frac{1}{\Lambda(n)} \cdot \left( \frac{1}{n-k}  + \frac{1}{k} - \frac{4}{n-1} \right).
\end{eqnarray*}
Thus, we get that overall
$\E_{x \sim \mu} [x_ix_j ]$ equals
\begin{eqnarray*}
 &&\frac{1}{ \Lambda(n)}  \left(1- \frac{4}{n} \right) \cdot \frac{2n}{n-1} + \sum_{k=2}^{n-2}\frac{1}{\Lambda(n)} \cdot \left( \frac{1}{n-k}  + \frac{1}{k} - \frac{4}{n-1} \right) \\
&=& \frac{1}{\Lambda(n)} \left(  2 + \frac{2}{n-1} - \frac{8}{n-1} \right) +\frac{1}{\Lambda(n)} \left( \sum_{k=2}^{n-2} \frac 1k + \frac{1}{n-k}  \right)  -\frac{4}{\Lambda(n)} + \frac{8}{\Lambda(n) (n-1)} \\
&=& \frac{1}{\Lambda(n)} \left( \sum_{k=1}^{n-1} Q(n,k) \right) -\frac{4}{\Lambda(n)} = 1 - \frac{4}{\Lambda(n)},
 \end{eqnarray*}
as was to be shown.
\end{proof}

\medskip

\begin{proof}[Proof of Lemma~\ref{lem:ortho}]
We begin by observing that $$\mathbf{E}_{x \sim \mu} [L_i (x) L_0(x) ]  = \mathbf{E}_{x \sim \mu} [L_i(x)] = \mathbf{E}_{x \sim \mu} [\alpha(x_1 + \ldots +x_n) + \beta x_i ] =0$$
since $\mathbf{E}_{x \sim \mu} [x_i]=0$.  Next, we solve for $\alpha$ and $\beta$ using the orthonormality conditions on the set $\{L_i \}_{i=1}^n$.
As $\mathbf{E}_{x \sim \mu} [L_i (x) L_j(x) ] =0$ and $\mathbf{E}_{x \sim \mu} [L_i(x) L_i(x)]=1$, we get that $\mathbf{E}_{x \sim \mu} [L_i(x) (L_i(x)-L_j(x))]=1$.  This gives
\begin{eqnarray*}
\mfed [L_i(x) \cdot (L_i(x) - L_j(x))] &=& \mfed [L_i(x) \cdot \beta (x_i -x_j) ]\\
& =& \mfed [ \beta ((\alpha + \beta) x_i + \alpha x_j) \cdot (x_i -x_j)]  \\
&=& \alpha \beta + \beta^2  - \alpha \beta  - \beta^2 \mfed [x_j x_i]  \\ &=& \beta^2 (1 - \mfed [x_ix_j]) = 4 \beta^2 / \Lambda(n) =1,
\end{eqnarray*}
where the penultimate equation above uses Proposition~\ref{prop:corr}. Thus, we have shown that  $\beta = \frac{\sqrt{\Lambda(n)}}{2}$.  To solve for $\alpha$, we note that
$$
\sum_{i=1}^n L_i(x) = (\alpha n + \beta)(x_1  + \ldots +x_n).
$$
However, since the set $\{L_i \}_{i=1}^n$ is orthonormal with respect to the distribution $\mu$, we get that
\begin{eqnarray*}
& \mfed [(L_1(x) + \ldots + L_n(x))(L_1(x) + \ldots +L_n(x)) ] =n
\end{eqnarray*}
and consequently
\begin{eqnarray*}
(\alpha n + \beta)^2 \ \mfed [ (x_1 + \ldots +x_n ) (x_1 + \ldots +x_n) ] =n
\end{eqnarray*}
Now, using Proposition~\ref{prop:corr}, we get
\begin{eqnarray*}
\mfed [ (x_1 + \ldots +x_n ) (x_1 + \ldots +x_n) ] &=& \sum_{i=1}^n \mfed [x_i^2] + \sum_{i \not =j } \mfed [x_i x_j]  \\
&=& n + n(n-1) \cdot \left( 1- \frac{4}{\Lambda(n)} \right)
\end{eqnarray*}
Thus, we get that
$$
(\alpha n + \beta)^2 \cdot \left( n + n(n-1) \cdot \left( 1- \frac{4}{\Lambda(n)} \right) \right)
= n.
$$
Simplifying further,
\begin{eqnarray*}
 (\alpha n + \beta) = \sqrt{\frac{\Lambda(n)}{n \Lambda(n)  - 4 (n-1)}}
\end{eqnarray*}
and thus
$$
\alpha = \frac{1}{n} \cdot \left( \sqrt{\frac{\Lambda(n)}{n \Lambda(n)  - 4 (n-1)}}  -\frac{\sqrt{\Lambda(n)}}{2} \right)
$$
as was to be shown.
\end{proof}

\medskip

\noindent {\bf Relating the Shapley values to the Fourier coefficients.}
\ignore{
Let us fix some useful notation.
For any $f : \{-1,1\}^n \rightarrow \mathbb{R}$ and any $i=0,\dots,n$,
let us define $\hat{f}(i) \eqdef \langle f, L_i \rangle_\mu = \mfed [ f(x) \cdot L_i(x)]$ and
$f^*(i) \eqdef  \langle f, x_i \rangle_\mu =  \mfed [ f(x) x_i]$.  We may view $\hat{f}(i)$ as the
``$i$-th Fourier coefficient of $f$ with respect to $\mu$.''
} The next lemma gives a useful expression for $\hat{f}(i)$ in terms 
of $\tilde{f}(i)$:

\begin{lemma}\label{lem:foushap}
Let $f : \{-1,1\}^n \rightarrow [-1,1]$ be any bounded function. Then for each
$i=1,\dots,n$ we have
%\begin{equation}\label{eq:2}
$$\hat{f}(i) =\frac{2\beta}{\Lambda(n)}\cdot  \left(\tilde{f}(i)  
-\frac{f(\mathbf{1}) - f(\mathbf{-1})}{n} \right) + \frac{1}{n}  
\cdot \littlesum_{j=1}^n \hat{f}(j).
$$
%\end{equation}
\end{lemma}

\begin{proof}
Lemma~\ref{lem:ortho} gives us that
$L_i(x) = \alpha (x_1 + \ldots + x_n) + \beta x_i$, and thus we have
\begin{eqnarray}\label{eq:exp1}
 \hat{f}(i) \equiv \mfed [f(x)\cdot L_i(x)] &=&
\alpha  \left(\sum_{j=1}^n \mfed [ f(x) \cdot x_j]
 \right)  + \beta \mfed [ f(x) \cdot x_i] \nonumber \\
& = & \alpha \sum_{j=1}^n f^*(j) + \beta f^*(i).
 \end{eqnarray}
 Summing this for $i=1$ to $n$, we get that
 \begin{eqnarray}
 \sum_{j=1}^n \hat{f}(j) = (\alpha n + \beta)  \sum_{j=1}^n f^*(j). \label{eq:summed}
 \end{eqnarray}
Plugging this into (\ref{eq:exp1}), we get that
 \begin{equation}\label{eq:exp2}
 f^{*}(i) = \frac 1\beta \cdot \left( \hat{f}(i) -\frac{\alpha}{\alpha n + \beta}\cdot \sum_{j=1}^n  \hat{f}(j) \right)
 \end{equation}
Now recall that from Lemma~\ref{lem:expshap}, we have
\begin{eqnarray*}
\tilde{f}(i) &=& \nu(f) + \frac{\Lambda(n)}{2} \cdot \left( \mathop{\mathbf{E}}_{x \sim \mu}\left[ f(x) \cdot x_i \right] - \mathop{\mathbf{E}}_{x \sim \mu} \left[ f(x) \cdot (\littlesum_{i=1}^n x_i) / n \right] \right)\\
&=& \nu(f) +\frac{\Lambda(n)}{2} \cdot \left(f^*(i) - \frac{ \sum_{j=1}^n f^*(j)}{n} \right)
\end{eqnarray*}
where $\nu(f) = (f(\mathbf{1}) - f(\mathbf{-1}))/n$.
 Hence, combining the above with (\ref{eq:summed}) and~(\ref{eq:exp2}), we get  $$
 \frac 1\beta \cdot \left( \hat{f}(i) -\frac{\alpha}{\alpha n + \beta}\cdot \sum_{j=1}^n \hat{f}(j) \right) = \frac{2}{\Lambda(n)} \cdot (\tilde{f}(i) -\nu(f)) + \frac 1{n(\alpha n + \beta)} \cdot \sum_{j=1}^n \hat{f}(j).
 $$
 From this, it follows that
 $$
 \frac 1\beta \cdot  \hat{f}(i) = \frac{2}{\Lambda(n)} \cdot  (\tilde{f}(i)  -\nu(f)) +\frac{1}{\alpha n + \beta} \cdot  \left( \frac{1}{n} + \frac \alpha\beta \right)\cdot \sum_{j=1}^n \hat{f}(j),
 $$
and hence
$$
\hat{f}(i) =\frac{2\beta}{\Lambda(n)}\cdot  (\tilde{f}(i) -\nu(f))  + \frac{1}{n}  \cdot \sum_{j=1}^n \hat{f}(j)
$$ as desired.
\end{proof}

\medskip

\noindent {\bf Bounding Shapley distance in terms of Fourier distance.}
Recall that the Shapley distance $\dshap(f,g)$ between $f,g: \{-1,1\}^n 
\to [-1,1]$ is defined as $\dshap(f,g) \eqdef \sqrt{\littlesum_{i=1}^n (\tilde{f}(i) - 
\tilde{g}(i))^2}.$ We define the \emph{Fourier distance} between $f$ and $g$ 
as
$\dfour(f,g) \eqdef
\sqrt{\littlesum_{i=0}^n (\hat{f}(i) - \hat{g}(i))^2}.
$

Our next lemma shows that if the Fourier distance between $f$ and $g$ is small then so is
the Shapley distance.
\begin{lemma}\label{lem:shapclose}
Let $f, g : \{-1,1\}^n \rightarrow [-1,1]$. Then,
$$
\dshap(f,g) \le \frac{4}{\sqrt{n}} + \frac{\Lambda(n)}{2\beta} \cdot \dfour(f,g).
$$
\end{lemma}

\begin{proof}
Let $\nu(f)= (f(\mathbf{1}) - f(\mathbf{-1}))/n$ and $\nu(g) =
(g(\mathbf{1}) - g(\mathbf{-1}))/n$.
From Lemma~\ref{lem:foushap}, we have that for all $1 \le i \le n$,
$$
\frac{\Lambda(n)}{2\beta} \cdot \left(\hat{f}(i)  - \frac{\sum_{j=1}^n \hat{f}(j)}{n}\right) + \nu(f) = \tilde{f}(i).
$$
Using a similar relation for $g$, we get that for every $1 \le i \le n$,
$$
\frac{\Lambda(n)}{2\beta} \cdot \left(\hat{f}(i)  - \frac{\sum_{j=1}^n
\hat{f}(j)}{n} - \hat{g}(i) + \frac{\sum_{j=1}^n \hat{g}(j)}{n} \right)
+ \nu(f) -\nu(g) = \tilde{f}(i) -\tilde{g}(i).
$$
We next define the following vectors:
let $v \in \mathbb{R}^n$ be defined by $v_i = \tilde{f}(i) - \tilde{g}(i)$, $i \in [n]$
(so our goal is to bound $\|v\|_2$).
Let $u \in \mathbb{R}^n$ be defined by $u_i = \nu(f) - \nu(g)$, $i \in [n]$. Finally,
let $w \in \mathbb{R}^n$ be defined by
$$
w_i = \left(\hat{f}(i)  - \frac{\sum_{j=1}^n \hat{f}(j)}{n} - \hat{g}(i) +
\frac{\sum_{j=1}^n \hat{g}(j)}{n} \right), \quad i \in [n].
$$
With these definitions the vectors $u$, $v$ and $w$ satisfy
$
\frac{\Lambda(n)}{2\beta} \cdot w + u = v$, and hence we have
\[
\Vert v \Vert_2 \le \Vert u \Vert_2 + \frac{\Lambda(n)}{2\beta} \cdot \Vert w
\Vert_2.
\]
Since the range of $f$ and $g$ is $[-1,1]$, we immediately have that
$$ \Vert u \Vert_2  = \left( \frac{f(\mathbf{1}) - g(\mathbf{1}) -
f(\mathbf{-1}) + g(\mathbf{-1})}{n}\right) \cdot \sqrt{n} \le
\frac{4}{\sqrt{n}},$$
so all that remains is to bound $\Vert w \Vert_2$ from above.
To do this, let us define another vector $w' \in \mathbb{R}^n$ by
$w'_i = \hat{f}(i) - \hat{g}(i)$. Let $\mathbf{e} \in \mathbb{R}^n$
denote the unit vector $\mathbf{e} = (1/\sqrt{n}, \ldots, 1/\sqrt{n})$.
Letting $w'_{\mathbf{e}}$ denote the projection of  $w$ along $\mathbf{e}$,
it is easy to see that
$$
w'_{\mathbf{e}} = \left( \frac{\sum_{j=1}^n (\hat{f}(j) - \hat{g}(j))}{n} , \ldots, \frac{\sum_{j=1}^n (\hat{f}(j) - \hat{g}(j))}{n} \right).
$$
This means that $w = w' - w'_{\mathbf{e}}$ and that $w$ is the projection
of $w'$ in the space orthogonal to $\mathbf{e}$.
Consequently we have $\Vert w \Vert_2 \le \Vert w' \Vert_2$, and hence
$$
\|v\|_2 \leq \frac{4}{\sqrt{n}} + \frac{\Lambda(n)}{2\beta} \|w'\|_2
$$
as was to be shown.
\end{proof}

\medskip

\noindent {\bf Bounding Fourier distance by ``correlation distance.''}
The following lemma will be useful for us since it lets
us bound from above Fourier distance in terms of the distance
between vectors of correlations with individual variables:

\begin{lemma} \label{lem:hatandstar}
Let $f,g: \{-1,1\}^n \rightarrow \mathbb{R}$. Then we have $$
\dfour(f,g) \le O(\sqrt{\log n}) \cdot
\sqrt{\littlesum_{i=0}^n (f^*(i) - g^*(i))^2}.$$
\end{lemma}

\begin{proof}
We first observe that $\hat{f}(0)=f^*(0)$ and $\hat{g}(0)=g^*(0)$, so
$(\hat{f}(0)-\hat{g}(0))^2 = (f^*(0)-g^*(0))^2$.  Consequently
it suffices to prove that
$$\sqrt{\sum_{i=1}^n (\hat{f}(i) - \hat{g}(i))^2} \le O(\sqrt{\log n}) \cdot
\sqrt{\sum_{i=1}^n (f^*(i) - g^*(i))^2},$$
which is what we show below.

From (\ref{eq:exp1}), we get
$$
\hat{f}(i) = \alpha  \sum_{j=1}^n f^*(j)  + \beta f^*(i) \
\quad \text{and} \quad  \hat{g}(i) = \alpha  \sum_{j=1}^n g^*(j)  + \beta
g^*(i).
$$
and thus we have
$$
(\hat{f}(i) - \hat{g}(i))  = \alpha \left(\sum_{j=1}^n f^*(j)
 - \sum_{j=1}^n g^*(j) \right) + \beta (f^*(i) - g^*(i)).
$$
Now consider vectors $u,v,w \in \mathbb{R}^n$ where for $i \in
[n]$, 
$$
u_i = (\hat{f}(i) - \hat{g}(i)), \quad v_i =  \left(\sum_{j=1}^n
f^*(j)  - \sum_{j=1}^n g^*(j) \right), \quad \text{and} \quad w_i = (f^*(i) -
g^*(i))
$$
By combining the triangle inequality and Cauchy-Schwarz, we have
$$
\Vert u \Vert_2^2 \le 2 (\alpha^2 \Vert v \Vert_2^2 + \beta^2
 \Vert w \Vert_2^2),
$$
and moreover
$$
\Vert v \Vert_2^2 = n \left(\sum_{j=1}^n f^*(j)  - \sum_{j=1}^n
g^*(j) \right)^2 \le n^2 \left(\sum_{j=1}^n (f^*(j)  -
 g^*(j))^2 \right) = n^2 \Vert w \Vert_2^2.
$$
Hence, we obtain
$$
\Vert u \Vert_2^2 \le 2 (\alpha^2 n^2 + \beta^2) \Vert w
\Vert_2^2 
$$
Recalling that $\alpha^2 n^2 = \Theta (\log n)$ and $\beta^2 =
\Theta (\log n)$, we conclude that
$$
\dfour(f,g) = \sqrt{\sum_{i=1}^n (\hat{f}(i) - \hat{g}(i))^2} \le O(\sqrt{\log n}) \cdot
\sqrt{\sum_{i=1}^n (f^*(i) - g^*(i))^2}
$$
which completes the proof.
\end{proof}

\medskip

\noindent 
{\bf From Fourier closeness to $\ell_1$-closeness.}
An important technical ingredient in our work is the notion of an affine
form $\ell(x)$ having ``good anti-concentration'' under distribution
$\mu$; we now give a precise definition to capture this.

\begin{definition}[Anti-concentration]
Fix $w \in \mathbb{R}^n$ and $\theta \in \mathbb{R}$, and let the affine 
form $\ell(x)$ be $\ell(x) \eqdef w \cdot x - \theta.$  We say that
$\ell(x)$ is \emph{$(\delta, \kappa)$-anti-concentrated under $\mu$} if
$
\Pr_{x \sim \mu} [|\ell(x)| \le \delta ] \le \kappa.
$
\end{definition}

The next lemma plays a crucial role in our results.  It essentially shows
that for $f = \sign (w \cdot x - \theta)$, if the affine form
$\ell(x) = w \cdot x -\theta$ is anti-concentrated,
then \emph{any} bounded function
$g : \{-1,1\}^n \rightarrow [-1,1]$ that has $\dfour(f,g)$ small
must in fact be close to $f$ in $\ell_1$ distance under $\mu$.

\begin{lemma} \label{lem:anticonc-and-ell1}
Let $f : \{-1,1\}^n \rightarrow \{-1,1\}$,
$f = \sign (w \cdot x -\theta)$ be such that
$w \cdot x - \theta$ is $(\delta,\kappa)$-anti-concentrated under
$\mu$ (for some $\kappa \leq 1/2$), where $|\theta| \leq \|w\|_1.$
Let $g : \{-1,1\}^n \rightarrow [-1,1]$ be such that
$\dfour(f,g) \leq \rho.$  Then we have
\[\mathbf{E}_{x \sim \mu} [ |f(x) - g(x)| ] \le (4 \Vert w \Vert_1 \sqrt{\rho})/\delta + \new{2} \kappa.\]
\end{lemma}

\begin{proof}
Let us rewrite $\ell(x) \eqdef w \cdot x -\theta$ as a linear
combination of the orthonormal basis elements $L_0,L_1,\dots,L_n$ (w.r.t.
$\mu$), i.e.,
\[
\ell(x) =   \hat{\ell}(\emptyset) L_0 + \sum_{i=1}^n \hat{\ell}(i) L_i.
\]
Recalling the definitions of $L_i$ for $i=1,\dots,n$ and the
fact that $L_0=1,$ we get $\hat{\ell}(\emptyset)=-\theta.$

We first establish an upper bound on $\theta^2 + \sum_{j=1}^n \hat{\ell}(j)^2 $
as follows :
\begin{eqnarray*}
\theta^2 + \sum_{j=1}^n \hat{\ell}(j)^2  = \mathbf{E}_{x \sim \mu} [(w \cdot x -\theta)^2]  &\le&  2\mathbf{E}_{x \sim \mu} [(w \cdot x)^2]  + 2 \theta^2  \\
&\le& 2 \Vert w \Vert_1^2 + 2 \Vert w \Vert_1^2 =  4 \Vert w \Vert_1^2.
\end{eqnarray*}
The first equality above uses the fact that the $L_i$'s are orthonormal
under $\mu$, while the first inequality uses
$(a+b)^2 \le 2 (a^2 + b^2)$ for $a,b \in \mathbb{R}$.
The second inequality uses the assumed bound on $|\theta|$ and the fact
that $|w \cdot x|$ is always at most $\Vert w \Vert_1$.

Next, \new{linearity of expectation} gives us that
\begin{eqnarray}
\mathbf{E}_{x \sim \mu} [(f(x) - g(x)) \cdot
(w \cdot x - \theta) ]  &=&  \theta (\hat{g}(0) - \hat{f}(0))
+ \sum_{j=1}^n \hat{\ell}(i)(\hat{f}(i) - \hat{g}(i))  \nonumber \\
&\le& \sqrt{\sum_{j=0}^n (\hat{f}(j) - \hat{g}(j))^2}
\cdot \sqrt{\theta^2 + \sum_{j=1}^n \hat{\ell}(i)^2} \nonumber \\ 
&\le& 2 \Vert w \Vert_1 \sqrt{\rho} \label{eq:planch}
\end{eqnarray}
where the first inequality is Cauchy-Schwarz and the second follows by the
conditions of the lemma.

Now note that since $f=\sign(w \cdot x- \theta)$, for all $x \in \bn$
we have $$(f(x) - g(x)) \cdot (w \cdot x -\theta)
= |f(x) - g(x)| \cdot |w \cdot x - \theta|$$
Let $E$ denote the event that $|w \cdot x - \theta| >\delta$.
Using the fact that the affine form $w \cdot x - \theta$ is
$(\delta, \kappa)$-anti-concentrated, we get that
$\Pr[E] \geq 1 - \kappa$, and hence
\begin{eqnarray*}\mathbf{E}_{x \sim \mu} [(f(x) - g(x))
\cdot (w \cdot x -\theta)] &\ge& \mathbf{E}_{x \sim \mu} [(f(x) - g(x))\cdot (w \cdot x - \theta) \ | \ E] \Pr[ E]
\\
&\ge& \delta (1-\kappa) \mathbf{E}_{x \sim \mu} [|f(x) - g(x)| \ | \ E].
\end{eqnarray*}
Recalling that $\kappa \leq 1/2$, this together with (\ref{eq:planch})
implies that
\[ \mathbf{E}_{x \sim \mu} [|f(x) - g(x)| \ | \ E]  \le \frac{4 \Vert w \Vert_1
\sqrt{\rho}}{\delta},\]
which in turn implies (since $|f(x) - g(x)| \le \new{2}$ for all $x \in \bn$) that
$$
\mathbf{E}_{x \sim \mu} [|f(x) - g(x) | ]  \le \frac{4 \Vert w \Vert_1 
\sqrt{\rho}}{\delta} + \new{2} \kappa
$$
as was to be shown.
\end{proof}

% %%%%%%%%%%%%%%%%%%%%%%%%%%%%%%%%%%%%%%%%%%%%%%%%%%%%%%%%%%%%%%%%%%%%%%%%%
% End of shapley-chow
% %%%%%%%%%%%%%%%%%%%%%%%%%%%%%%%%%%%%%%%%%%%%%%%%%%%%%%%%%%%%%%%%%%%%%%%%%

\section{A Useful Anti-concentration Result} \label{sec:str}

In this section we prove an anti-concentration result for monotone
increasing $\eta$-reasonable affine forms (see
Definition~\ref{def:reasonable})
under the distribution $\mu$.
Note that even if $k$ is a constant
the result gives an anti-concentration probability of 
$O(1/\log n)$; this will be crucial in the proof of our
first main result in Section~\ref{sec:mainresults}.

\begin{theorem} \label{thm:ac-general-weights}
Let $L(x) = w_0+\littlesum_{i=1}^n w_i x_i$ be a monotone
increasing $\eta$-reasonable affine form, so
$w_i \geq 0$ for $i \in [n]$ and $|w_0| \leq (1-\eta) \littlesum_{i=1}^n
|w_i|$. Let $k \in [n],0<\zeta<1/2$, $k \ge 2/\eta$ and $r \in \R_{+}$
be such that $|S| \geq k$, where $ S: = \{ i \in [n] : |w_i| \geq r \}$.
Then
\[ \Pr_{x \sim \mu} \left[| L(x)| < r \right]  = O\left( \frac{1}{\log n} \cdot \frac{1}{k^{1/3- \zeta}} \cdot \left( \frac{1}{\zeta} + \frac{1}{\eta}\right)\right) .\]
\end{theorem}

This theorem essentially says that under the distribution $\mu$, 
the random variable $L(x)$ falls in the interval $[-r,r]$ with only a
very small probability.  Such theorems are known in the literature as 
``anti-concentration" results, but almost all such results are for the 
uniform distribution or for other product distributions, 
and indeed the proofs of such results typically crucially use the fact 
that the distributions are product distributions. 

In our setting, the distribution $\mu$ is not even a pairwise 
independent distribution, so standard approaches for proving
anti-concentration cannot be directly applied.
Instead, we exploit the fact that $\mu$ is a \emph{symmetric} distribution;
a distribution is symmetric if the probability mass it assigns to an 
$n$-bit string $x \in \bn$ depends only on the number of $1$'s of $x$ 
(and not on their location within the string).  This enables
us to perform a somewhat delicate reduction to known anti-concentration
results for biased product distributions.
Our proof adopts a point of view which is 
inspired by the combinatorial proof of the basic Littlewood-Offord theorem 
(under the uniform distribution on the hypercube) due to 
Benjamini et.~al.~\cite{BKS:99}.  The detailed
proof is given in the following subsection.

\subsection{Proof of Theorem~\ref{thm:ac-general-weights}.}
Recall that $\bn_{=i}$ denotes the $i$-th ``weight level'' of the
hypercube, i.e., $\{ x \in \bn: \weight(x) = i\}$.
We view a random draw $x \sim \mu$ as being done according to a two-stage process:
\begin{enumerate}
\item Draw $i \in [n-1]$ with probability $q(n, i) \eqdef Q(n,i) / \Lambda(n).$ 
(Note that this is the probability $\mu$ assigns to $\bn_{=i}$.) 
\item Independently pick a uniformly random permutation $\pi: [n] \to [n]$, i.e., $\pi \sim_R \ms$. 
The string $x$ is defined to have $x_{\pi(1)} = \ldots = x_{\pi(i)} =1$ and $x_{\pi(i+1)} = \ldots = x_{\pi(n)} =-1.$
\end{enumerate}
It is easy to see that the above description of $\mu$ is equivalent to its original definition.  
Another crucial observation is that any symmetric distribution can be sampled in the same way, 
with $q(n,k)$ being the only quantity dependent on the particular distribution. 
We next define a $(r, i)$-balanced permutation. 
\begin{definition}[$(r, i)$-balanced permutation]
A permutation $\pi: [n]\to [n]$ is called $(r, i)$-balanced if $| w_0 + \littlesum_{j=1}^i w_{\pi(j)} -  \littlesum_{j=i+1}^n w_{\pi(j)} | \leq r.$
\end{definition}

For $i \in [n-1]$, let us denote by $p(r, i)$ the fraction of all $n!$ permutations that are $(r, i)$ balanced. 
That is, 
\[ p(r, i) = \Pr_{\pi \sim_R \ms} \left[ | w_0 + \littlesum_{j=1}^i w_{\pi(j)} -  \littlesum_{j=i+1}^n w_{\pi(j)} | \leq r \right] .\]
At this point, as done in~\cite{BKS:99}, we use the above two-stage process defining $\mu$ to express the desired ``small ball'' probability in a more convenient way.
Conditioning on the event that the $i$-th layer is selected in the first stage, the probability that  $|L(x)| < r$ is $p(r, i)$. By the law of total probability we can write: 
\[ \Pr_{x \sim \mu} \left[ |L(x)| < r \right]  = \littlesum_{i=1}^{n-1} p(r, i) q(n, i).\]
We again observe that $p(r,i)$ is only dependent on the affine form $L(x)$ and 
does not depend on the particular symmetric distribution; 
$q(n,i)$ is the only part dependent on the distribution. The high-level 
idea of bounding the quantity $\littlesum_{i=1}^{n-1} p(r, i) q(n, i)$ is as 
follows: For $i$ which are ``close to $1$ or $n-1$'', we use Markov's 
inequality to argue that the corresponding $p(r,i)$'s are suitably small; 
this allows us to bound the contribution of these indices to the sum, using 
the fact that each $q(n,i)$ is small. For the remaining $i$'s, we use the 
fact that the $p_i$'s are identical for all symmetric distributions. This 
allows us to perform a subtle ``reduction'' to known anti-concentration results 
for biased product distributions.
%We bound this quantity from above as follows:  For indices $i$ ``close to $1$ or $n-1$'', we use an averaging argument.
%For the remaining indices, we exploit the fact that the distribution $\mu$ is symmetric, which implies that the $p_i$'s are distribution independent over such distributions. 
%This allows to perform a subtle ``reduction'' to known anti-concentration results for biased product distributions.
%The crucial point is that we use both the version of L-O for biased product distributions and the restriction lemma to get 
%a result for our weight setting. And then we use Anindya's trick to get it for the distribution $\mu$

We start with the following simple claim, a consequence of Markov's inequality, that shows that if one of $i$ or $n-i$ is reasonably small,
the probability $p(r, i)$ is quite small.

\begin{myclaim} \label{claim:markov}
For all $i \in [n-1]$ we have
\[p(r, i) \leq (4/\eta) \cdot \min\{i, n-i \}/n.\]
\end{myclaim}

\begin{proof}
For $i \in [n-1]$, let $\mathcal{E}_i = \{ \pi \in \ms :  | w_0 + \littlesum_{j=1}^i w_{\pi(j)} -  \littlesum_{j=i+1}^n w_{\pi(j)} | \leq r \}.$
By definition we have that $p(r, i) = \Pr_{\pi \sim_R \ms} \left[ \mathcal{E}_i \right]$.

Let $i \leq n/2.$ If the event $\mathcal{E}_i$ occurs, we certainly have that  $w_0 + \littlesum_{j=1}^i w_{\pi(j)} -  \littlesum_{j=i+1}^n w_{\pi(j)} \geq -r$
which yields that
\[  \littlesum_{j=1}^i w_{\pi(j)} \geq (1/2) (\littlesum_{i=1}^n w_i  - r - w_0).\]
That is, \[ p(r,i) \leq \Pr_{\pi \sim_R \ms} \left[ \littlesum_{j=1}^i w_{\pi(j)} \geq (1/2) (\littlesum_{i=1}^n w_i  - r - w_0) \right]. \]
Consider the random variable $X =  \littlesum_{j=1}^i w_{\pi(j)}$ and denote $\alpha \eqdef  (1/2) (\littlesum_{i=1}^n w_i  - r - w_0)$. 
We will bound from above the probability
 \[ \Pr_{\pi \sim_R \ms} \left[ X \geq \alpha \right].\]
Since $\pi$ is chosen uniformly from $\ms$, we have that $\E_{\pi \sim_R \ms}[w_{\pi(j)}]  =(1/n) \cdot \littlesum_{i=1}^n w_i $, hence 
$$\E_{\pi \sim_R \ms}[X] = (i/n) \cdot \littlesum_{i=1}^n w_i.$$
Recalling that $|w_0| \leq (1-\eta) \cdot \littlesum_{i=1}^n w_i$ and 
noting that $\littlesum_{i=1}^n w_i  \geq \littlesum_{i \in S} w_i  \geq kr \geq (2/\eta) \cdot r$, we get
\[ \alpha \geq (\eta/4)  \cdot \littlesum_{i=1}^n w_i .\]
Therefore, noting that $X>0$, by Markov's inequality, we obtain that
\[  \Pr_{\pi \sim_R \ms} \left[ X \geq \alpha \right] \leq \frac{\E_{\pi \sim_R \ms}[X] }{\alpha} \leq (4/\eta) \cdot (i/n)  \]
as was to be proven.

If $i \geq n/2$, we proceed analogously. If $\mathcal{E}_i$ occurs, we have
$w_0 + \littlesum_{j=1}^i w_{\pi(j)} -  \littlesum_{j=i+1}^n w_{\pi(j)} \leq r$
which yields that
\[  \littlesum_{j=i+1}^n w_{\pi(j)} \geq (1/2) (\littlesum_{i=1}^n w_i + w_0 - r).\]
We then repeat the exact same Markov type argument for the random variable $\littlesum_{j=i+1}^n w_{\pi(j)}$. 
This completes the proof of the claim.
\end{proof}

Of course, the above lemma is only useful when either $i$ or $n-i$ is relatively small. Fix $i_0 < n/2$ (to be chosen later).
Note that, for all $i \leq n/2$, it holds $q(n, i) \leq \frac{2}{i \cdot \Lambda(n)}$. By Claim~\ref{claim:markov} we thus get
that 
\begin{equation}\label{lb:1}\littlesum_{i = 1}^{i_0}  p(r, i) q(n,i) \leq \littlesum_{i=1}^{i_0}  \frac{2}{i \cdot \Lambda(n)} \cdot \frac{4}{\eta} \cdot \frac{i}{n} 
                                                                   \leq \frac{8 i_0}{\eta \cdot n \cdot  \Lambda(n)}.\end{equation}
By symmetry, we get \begin{equation}\label{lb:2}\littlesum_{i = n-i_0}^{n-1}  p(r, i) q(n,i) \leq \frac{8 i_0}{\eta \cdot n \cdot  \Lambda(n)}.\end{equation}

We proceed to bound from above the term  $\littlesum_{i = i_0+1}^{n-i_0-1}  p(r, i) q(n,i)$.
To this end, we exploit the fact, mentioned earlier,
that the $p(r, i)$'s depend only on the affine form
and not on the particular symmetric distribution over
weight levels. We use a subtle argument to essentially reduce
anti-concentration statements about $\mu$ to known anti-concentration results.

For $\delta \in (0,1)$ let $D_{\delta}$ be the $\delta$-biased distribution over $\{-1,1\}^n$;  that is the product distribution in which each coordinate 
is $1$ with probability $\delta$ and $-1$ with probability $1-\delta$. Denote by $g(\delta,i)$ the probability that $D_{\delta}$ assigns to $\bn_{=i}$, i.e.,
$g(\delta,i) =  \binom{n}{i} \delta^{i} (1-\delta)^{n-i}$. Theorem 5 now yields %for any weight-vector $w$ satisfying the condition of the theorem statement
$$
\Pr_{x \sim D_{\delta}} \left[   |L(x)| < r \right] \leq \frac{1}{\sqrt{k \delta(1-\delta)}}.
$$
Using symmetry, we view a random draw $ x \sim D_{\delta}$ as a two-stage procedure, exactly as in $\mu$, the only difference being that
in the first stage we pick the $i$-th weight level of the hypercube, 
$i \in [0,n]$, with probability $g(\delta, i)$. We can therefore write
$$\Pr_{x \sim D_{\delta}} \left[   |L(x)| < r  \right]  = \littlesum_{i=0}^{n} g(\delta,i) p(r, i)$$
and thus conclude that
\begin{equation} \label{eqn:lo}
\littlesum_{i=i_0+1}^{n-i_0-1} g(\delta,i) p(r, i) \leq \littlesum_{i=0}^{n} g(\delta,i) p(r,i) \leq \frac{1}{\sqrt{k \delta(1-\delta)}}.
\end{equation}
We now state and prove the following crucial lemma. The idea of the lemma is to bound from above the sum 
$\littlesum_{i = i_0+1}^{n-i_0-1}  p(r, i) q(n,i)$ by suitably averaging over anti-concentration bounds 
obtained from the $\delta$-biased product distributions:
\begin{lemma} \label{lem:reduction}
Let $F:[0,1] \to \R_+$ be such that $q(n,i) \leq  \int_{\delta=0}^1 F(\delta) g(\delta,i) d \delta$ for all $i \in [i_0+1, n-i_0-1].$
Then, \[ \littlesum_{i = i_0+1}^{n-i_0-1}  p(r, i) q(n,i) \leq \frac{1}{\sqrt{k}} \cdot \int_{\delta=0}^1 \frac{F(\delta)}{\sqrt{\delta(1-\delta)}} d\delta.\]
\end{lemma}
\begin{proof}
We have the following sequence of inequalities 
\begin{eqnarray*}
 \littlesum_{i = i_0+1}^{n-i_0-1}  p(r, i) q(n,i)   
                    &\leq&  \littlesum_{i = i_0+1}^{n-1-i_0} \left( \int_{\delta=0}^1 F(\delta) g(\delta,i) d \delta \right) \cdot p(r, i) \\
                     &=& \int_{\delta=0}^1 F(\delta) \left(\littlesum_{i = i_0+1}^{n-i_0-1} g(\delta,i) p(r,i) \right) d\delta \\ 
                     &\leq& \frac{1}{\sqrt{k}} \cdot \int_{\delta=0}^1 \frac{F(\delta)}{\sqrt{\delta(1-\delta)}} d\delta
\end{eqnarray*}
where the first line follows from the assumption of the lemma, the second uses linearity and the third uses (\ref{eqn:lo}).
\end{proof}

We thus need to choose appropriately a function $F$ satisfying the lemma 
statement which can give a non-trivial bound on the desired sum. 
Fix $\zeta > 0$, and define $F(\delta)$ as  
$$F(\delta) \eqdef
\frac{1024}{\Lambda(n)} \cdot \frac{(n+1)^{1/2+\zeta}}{ i_0^{1/2+\zeta}} 
\left( \frac{1}{\delta^{1/2-\zeta}} + \frac{1}{(1-\delta)^{1/2- \zeta}}
\right).$$ 
The following claim (proved in Section~\ref{sec:proofoftechnical})
says that this choice of $F(\delta)$ satisfies the 
conditions of Lemma~\ref{lem:reduction}:

\begin{myclaim} \label{claim:technical}
For the above choice of $F(\delta)$ and $i_0 \le i \le n-i_0$, $q(n,i) \le 
\int_{\delta=0}^1 F(\delta) g(\delta,i) d \delta$. 
\end{myclaim}

Now, applying Lemma~\ref{lem:reduction}, for this 
choice of $F(\delta)$, we get that 
 \begin{eqnarray*} 
&&\littlesum_{i = i_0+1}^{n-i_0-1}  p(r, i) q(n,i)\\
 &\leq& \frac{1}{\sqrt{k}} \cdot \frac{1024}{\Lambda(n)} \cdot \frac{(n+1)^{1/2+\zeta}}{i_0^{1/2+ \zeta}}\int_{\delta=0}^1 \left( \frac{1}{\delta^{1/2-\zeta}} + \frac{1}{(1-\delta)^{1/2- \zeta}}\right)\frac{1}{\sqrt{\delta(1-\delta)}} d\delta. \\
 &=&  O\left(\frac{1}{\zeta}  \cdot \frac{1}{\sqrt{k}} \cdot \frac{1}{\Lambda(n)} \cdot \frac{(n+1)^{1/2+\zeta}}{i_0^{1/2+ \zeta}}\right).
 \end{eqnarray*}

We choose (with foresight) $i_0 = \new{\lceil} \frac{n}{k^{1/3}} \new{\rceil}$.  Then the above expression simplifies to 
$$
\littlesum_{i = i_0+1}^{n-i_0-1}  p(r, i) q(n,i)  = O\left(\frac{1}{\zeta} \cdot \frac{1}{\Lambda(n)} \cdot \frac{1}{k^{1/3 -\zeta}}\right)
$$
Now plugging  $i_0 =  \new{\lceil} \frac{n}{k^{1/3}} \new{\rceil}$ in (\ref{lb:1}) and (\ref{lb:2}), we get 
$$
\sum_{i \le i_0 \vee i \ge n-i_0} p(r, i) q(n,i)  = O\left( \frac{1}{\eta \Lambda(n)}  \cdot \frac{1}{k^{1/3}}\right)
$$
Combining these equations, we get the final result,
and Theorem~\ref{thm:ac-general-weights} is proved. \qed

\subsection{Proof of Claim~\ref{claim:technical}}
\label{sec:proofoftechnical}

We will need the following basic facts : 
\begin{fact}\label{fac:gamma}
For $x, y \in \R_+$ let $\Gamma : \R_+ \rightarrow \R$ be the usual ``Gamma" function, so that
$$
\int_{\delta=0}^1 \delta^x (1-\delta)^y d\delta = \frac{\Gamma(x+1) \cdot \Gamma(y+1)}{\Gamma(x+y+2)}
$$
Recall that for $z \in \Z_{+}$, $\Gamma(z) = (z-1)!$. 
\end{fact}
\begin{fact}\label{fac:gammaapx} (Stirling's approximation)
\new{For $z \in \R_{+}$,} we have $
\Gamma (z) = \sqrt{\frac{2 \pi}{z}}  \cdot \left(\frac{z}{e}\right)^z \cdot \left( 1+ O\left(\frac{1}{z} \right) \right).
$
In particular, there is an absolute constant $c_0 > 0$ 
such that for $z \ge c_0$
$$
\frac 12 \cdot \sqrt{\frac{2 \pi}{z}}  \cdot \left(\frac{z}{e}\right)^z 
\le \Gamma (z)  \le 2 \cdot \sqrt{\frac{2 \pi}{z}}  \cdot \left(\frac{z}{e}\right)^z .
$$
\end{fact}
\begin{fact}\label{fac:limit}
For $x \in \mathbb{R}$ and $x \ge 2$, we have
$
\left( 1- \frac{1}{x} \right)^{x} \ge \frac{1}{4}.
$
\end{fact}

We can now proceed with the proof of Claim~\ref{claim:technical}.
We consider the case when $i_0 \le i \le n/2$. (The proof of the complementary case ($n-i_0-1 \ge i >n/2$) is essentially identical.)
We have the following chain of inequalities:
\begin{eqnarray*}
&&\int_{\delta=0}^1 F(\delta) g(\delta,i) d \delta\\
 &=&       
\frac{1024}{\Lambda(n)}\cdot \frac{(n+1)^{1/2+\zeta}}{ i_0^{1/2+\zeta}} 
\cdot \binom{n}{i} \cdot \int_{\delta=0}^1 \delta^i (1-\delta)^{n-i} \cdot 
\left( \frac{1}{\delta^{1/2-\zeta}} + \frac{1}{(1-\delta)^{1/2- \zeta}}
\right)\ d\delta  \\ 
&\ge&  \frac{1024}{\Lambda(n)}\cdot \frac{(n+1)^{1/2+\zeta}}{
i_0^{1/2+\zeta}} \cdot \binom{n}{i} \cdot \int_{\delta=0}^1 \delta^{i-1/2 
+ \zeta} (1-\delta)^{n-i} \ d\delta  \\ 
&=& \frac{1024}{\Lambda(n)} \cdot \frac{(n+1)^{1/2+\zeta}}{ i_0^{1/2+\zeta}}
\cdot \binom{n}{i} \cdot  \frac{\Gamma(n-i+1) \cdot \Gamma(i+1/2 + \zeta)}
{\Gamma(n+3/2 + \zeta)}  \quad \textrm{(using Fact~\ref{fac:gamma})}\\ 
&=& \frac{1024}{\Lambda(n)}  \cdot \frac{(n+1)^{1/2+\zeta}}{ i_0^{1/2+\zeta}}
\cdot \frac{\Gamma(n+1)}{\Gamma(i+1) \cdot \Gamma(n-i+1)} \cdot  
\frac{\Gamma(n-i+1) \cdot \Gamma(i+1/2 + \zeta)}{\Gamma(n+3/2 + \zeta)} \\ 
&=& \frac{1024}{\Lambda(n)} \cdot \frac{(n+1)^{1/2+\zeta}}{ i_0^{1/2+\zeta}}
\cdot \frac{\Gamma(n+1) \cdot  \Gamma(i+1/2 + \zeta)}{\Gamma(i+1) \cdot 
\Gamma(n+3/2 + \zeta)}  \\
\end{eqnarray*}
We now proceed to bound from below the right hand side of the last inequality. Towards that, using Fact~\ref{fac:gammaapx} and assuming $n$ 
and $i$ are large enough, we have
\begin{eqnarray*}
&& \frac{\Gamma(n+1) \cdot  \Gamma(i+1/2 + \zeta)}{\Gamma(i+1) \cdot \Gamma(n+3/2 + \zeta)} \\ 
&\ge& \frac{1}{16} \cdot  \frac{(n+1)^{n+1/2}}{(i+1)^{i+1/2}} \cdot \frac{(i+1/2 + \zeta)^{i+\zeta}}{(n+3/2 +\zeta)^{n+\zeta+1}}   \\
 &\ge& \frac{1}{16} \cdot \frac{1}{n+2} \cdot  \frac{(n+1)^{n+1/2}}{(i+1)^{i+1/2}} \cdot \frac{(i+1/2 + \zeta)^{i+\zeta}}{(n+3/2 +\zeta)^{n+\zeta}} \\
 &\ge& \frac{1}{16} \cdot \frac{1}{n+2} \cdot  \frac{(n+1)^{n+\zeta}}{(n+3/2 +\zeta)^{n+\zeta}} \cdot \frac{(i+1/2 + \zeta)^{i+\zeta}}{(i+1)^{i+\zeta}} \cdot \frac{(n+1)^{1/2-\zeta}}{(i+1)^{1/2-\zeta}} \\
 &\ge& \frac{1}{256} \cdot  \frac{1}{n+2} \cdot \frac{(n+1)^{1/2-\zeta}}{(i+1)^{1/2-\zeta}}  \ge \frac{1}{512} \cdot \frac{1}{(n+1)^{1/2+\zeta}}\cdot \frac{1}{(i+1)^{1/2-\zeta}} 
% &\ge&  \frac{1}{16}\left( \frac{n+1}{i+1}\right)^{1/2-\zeta} \cdot  \left( 1 - \frac{1/2-\zeta}{i+1} \right)^{i+1/2 + \zeta}  \\
% &\ge& \frac{1}{64} \left( \frac{n+1}{i+1}\right)^{1/2-\zeta}  \quad \textrm{Using Fact~\ref{fac:limit}}
\end{eqnarray*}
Plugging this back, we get 
\begin{eqnarray*}
\int_{\delta=0}^1 F(\delta) g(\delta,i) d \delta  &\ge& \frac{1024}{\Lambda(n)} \cdot \frac{(n+1)^{1/2+\zeta}}{ i_0^{1/2+\zeta}}\cdot \frac{\Gamma(n+1) \cdot  \Gamma(i+1/2 + \zeta)}{\Gamma(i+1) \cdot \Gamma(n+3/2 + \zeta)} \\ &\ge&\frac{1024}{\Lambda(n)} \cdot \frac{(n+1)^{1/2+\zeta}}{ i_0^{1/2+\zeta}}\cdot  \frac{1}{512} \cdot \frac{1}{(n+1)^{1/2+\zeta}}\cdot \frac{1}{(i+1)^{1/2-\zeta}}  \\
&=& \frac{2}{\Lambda(n)}  \cdot \frac{1}{ i_0^{1/2+\zeta} } \cdot \frac{1}{(i+1)^{1/2-\zeta}} \ge \frac{2}{\Lambda(n)} \cdot \frac{1}{i} \ge q(n,i)
\end{eqnarray*} 
which concludes the proof of the claim. \qed

\section{A Useful Algorithmic Tool}
\label{sec:algorithmic}

In this section we describe a useful algorithmic tool 
arising from recent work in computational complexity theory.
The main result we will need is
the following theorem of \cite{TTV:09short} (the ideas go 
back to \cite{Impagliazzo:95short}
and were used in a different form in \cite{DDFS12}):

\begin{theorem}\label{TTV}(\cite{TTV:09short})
Let $X$ be a finite domain, $\mu$ be a samplable probability distribution over $X$, $f : X \rightarrow [-1,1]$ be a bounded function, and $\mathcal{L}$ be a finite family of Boolean functions $\ell: X \rightarrow \{-1,1\}$. There is an algorithm \mytextsf{Boosting-TTV} with the following
 properties:  Suppose \mytextsf{Boosting-TTV} is given as input a list $(a_{\ell})_{\ell \in \mathcal{L}}$ of real values and a parameter $\xi>0$ such that
$|\E_{x \sim \mu} [f(x) \ell(x)] - a_{\ell}| \le \xi/16$ for every $\ell \in \mathcal{L}$.
Then \mytextsf{Boosting-TTV} outputs a function $h : X \rightarrow [-1,1]$ with the following properties:

\begin{enumerate}
\item [(i)] $|\E_{x \sim \mu} [\ell(x) h(x) - \ell(x) f(x)]| \le \xi$ for every $\ell \in {\cal L}$;
\item [(ii)] $h(x)$ is of the form $h(x) = P_1({\frac \xi 2} \cdot 
\sum_{\ell \in \mathcal{L}} w_{\ell} \ell(x))$ where the $w_\ell$'s are integers whose absolute values sum to $O(1/\xi^2).$
\end{enumerate}
The algorithm runs for $O(1/\xi^2)$ iterations, where in each iteration it estimates
$\E_{x \sim \mu}[h'(x)\ell(x)]$ to within additive accuracy $\pm \xi/16$. Here
each $h'$ is a function of the form
$h'(x)=P_1( {\frac \xi 2} \cdot \sum_{\ell \in \mathcal{L}} v_{\ell} \ell(x))$, where the $v_\ell$'s are
integers whose absolute values sum to $O(1/\xi^2).$
\end{theorem}

We note that Theorem~\ref{TTV} is not explicitly stated in the above form
in \cite{TTV:09short}; in particular, neither the time complexity of the 
algorithm nor the fact that it suffices for the algorithm to be given 
``noisy" estimates $a_\ell$ of the values 
$\E_{x \sim \mu} [f(x) \ell(x)]$ is explicitly stated in \cite{TTV:09short}.  
So for the sake of completeness, in the following
 we state the algorithm in full (see Figure~\ref{fig:TTV}) 
and sketch a proof of correctness of this algorithm 
using results that are explicitly proved in \cite{TTV:09short}.

\begin{figure}[tb]
\label{fig:TTV}
\hrule
\vline
\begin{minipage}[t]{0.98\linewidth}
\vspace{10 pt}
\begin{center}
\begin{minipage}[h]{0.95\linewidth}
{\small

\underline{\mytextsf{Boosting-TTV}}\\ 

\medskip

\underline{\bf{Parameters:}}

\begin{tabular}{ccl}
$\xi$ &:=& positive real number \\
$\mu$ &:=& samplable distribution over finite domain $X$ \\
$\mathcal{L}$ &:=& finite list of functions such that all $\ell \in \mathcal{L}$ maps $X$ to $\{-1,1\}$. \\
$(a_{\ell})_{\ell \in \mathcal{L}}$ &:=& list 
of real numbers with the promise that some $f:  X \rightarrow [-1,1]$ has\\
& & $|\mathbf{E}_{x \sim \mu} [f(x) \ell(x)] - a_{\ell}| \le \xi/16$ for all
$\ell \in {\cal L}.$
\end{tabular}

\medskip

\underline{\bf{Output:}}

\noindent
An LBF $h(x) \equiv P_1( \sum_{\ell \in \mathcal{L}} w_{\ell} \ell(x))$, 
where $w_{\ell} \in \Z$, such that 
$\E_{x \sim \mu} [h(x) \ell(x)] - f(x)\ell(x)| \le \xi$ for all $\ell \in {\cal 
L}.$

\medskip

\underline{\bf{Algorithm}:} 
\begin{enumerate}
\item Let $\mathcal{L}^{0} \eqdef  \{\ell : \ell \in \mathcal{L}$  or $-\ell \in \mathcal{L}\}$.  Fix $\gamma \eqdef \xi/2.$
\item Let $h_0 \eqdef 0$.  Set $t=0$.
\item For each $\ell \in \mathcal{L}$, find $a_{\ell,t} \in \R$ such 
that $|\mathbf{E}_{x \sim \mu} [h_t(x) \ell(x)] - a_{\ell,t}| \le \xi/16$.
\item If $|a_{\ell} - a_{\ell,t}| \le \gamma$ for all $\ell \in {\cal L}$, 
then stop and output $h_t$.  
Otherwise, fix $\ell$ to be any element of $\mathcal{L}$ 
such that $|a_{\ell} - a_{\ell,t}| > \gamma$.
\begin{itemize}
\item If $a_{\ell} - a_{\ell,t} > \gamma $ then set $f_{t+1} \eqdef \ell$ else 
set $f_{t+1} \eqdef -\gamma$.  Note that $f_{t+1} \in \mathcal{L}^0$.
\item Define $h_{t+1}$ as $h_{t+1}(x) \eqdef P_1( \gamma (\littlesum_{j=1}^{t+1} f_j(x)))$.
\end{itemize}
\item Set $t = t+1$ and go to Step~3.
\end{enumerate}

\vspace{5 pt}

}
\end{minipage}
\end{center}

\end{minipage}
\hfill \vline
\hrule
\caption{Boosting based algorithm from \cite{TTV:09short} }
\end{figure}

\begin{proof}[{\bf {\em Proof of Theorem~\ref{TTV}}}]
It is clear from the description of the algorithm that (if and) when the 
algorithm \mytextsf{Boosting-TTV} terminates, the output $h$ satisfies
property (i) and has the form $h(x) = P_1( {\frac \xi 2}
\cdot \sum_{\ell \in \mathcal{L}} w_{\ell} \ell(x))$ where each $w_{\ell}$
is an integer.  It remains to bound the number of iterations (which
gives a bound on the sum of magnitudes of $w_\ell$'s) and indeed to show
that the algorithm terminates at all.

Towards this, we recall
Claim~3.4 in \cite{TTV:09short} states the following:
\begin{myclaim} \label{claim:my}
For all $x \in \mathop{supp}(\mu)$ and all $t \ge 1$, 
we have $\sum_{j=1}^t f_{j}(x) \cdot (f(x) - h_{j-1}(x)) 
\le (4/\gamma) + (\gamma t)/2.$
\end{myclaim}
We now show how this immediately gives Theorem~\ref{TTV}.
Fix any $j \ge 0$, and suppose without loss of generality that  
$a_{\ell} - a_{\ell, j} > \gamma$.  We have that 
$$|\mathbf{E}_{x \sim \mu}[f_{j+1}(x)h_{j}(x) ]  - a_{\ell, j} | \le \xi/16 
\quad \text{and hence} \quad \mathbf{E}_{x \sim \mu}[f_{j+1}(x) h_j(x)] \le 
a_{\ell, j} + \xi/16,$$
and similarly
$$
|\mathbf{E}_{x \sim \mu}[f_{j+1}(x) f(x)]  - a_{\ell} | \le \xi/16 \quad 
\text{and hence} \quad \mathbf{E}_{x \sim \mu}[f_{j+1}(x) f(x)] \ge a_{\ell}
- \xi/16.
$$
Combining these inequalities with $a_\ell - a_{\ell,j} > \gamma=\xi/2$, we
conclude that $$\mathbf{E}_{x \sim \mu} [f_{j+1}(x)(f(x) - h_j(x))] 
\ge 3\xi/8.$$  Putting this together with Claim~\ref{claim:my}, we get that
$$
\frac{3\xi t}{8} \le \sum_{j=1}^t \mathbf{E}_{x \sim \mu} [f_{j}(x)(f(x) 
- h_{j-1}(x))] \le \frac{4}{\gamma} + \frac{\gamma t}{2}.
$$
Since $\gamma = \xi/2$, this means that if the algorithm runs for $t$ time 
steps, then $8/\xi \ge (\xi t)/8$, which implies that $t \le 64/\xi^2$.  
This concludes the proof.
\end{proof}

% %%%%%%%%%%%%%%%%%%%%%%%%%%%%%%%%%%%%%%%%%%%%%%%%%%%%%%%%%%%%%%%%%%%%%%%%%
% End of algorithmic
% %%%%%%%%%%%%%%%%%%%%%%%%%%%%%%%%%%%%%%%%%%%%%%%%%%%%%%%%%%%%%%%%%%%%%%%%%

\section{Our Main Results}
\label{sec:mainresults}

In this section we combine ingredients from the previous subsections and prove 
our main results, Theorems~\ref{thm:main-arbitrary} and~\ref{thm:main-bounded}.

Our first main result gives an algorithm that works if \emph{any} monotone 
increasing $\eta$-reasonable LTF has approximately the right Shapley values:

\begin{theorem} \label{thm:main-arbitrary}
%Fix any (small) absolute constant $\xi > 0.$
There is an algorithm \mytextsf{IS} (for \mytextsf{Inverse-Shapley}) with the 
following properties.  \mytextsf{IS} is given as input an accuracy parameter 
$\eps > 0$, a confidence parameter $\delta > 0$, and $n$ real values 
$a(1),\dots, a(n)$; its output is a pair
$v \in \R^n, \theta \in \R.$  Its running time is 
$\poly(n,2^{\poly(1/\eps)},\log(1/\delta)).$
The performance guarantees of \mytextsf{IS} are the following:

\begin{enumerate}

\item Suppose there is a monotone increasing $\eta$-reasonable LTF $f(x)$
such that $\dshap(a,f) \leq 1/\poly(n,2^{\poly(1/\eps)})$.   
Then with probability $1-\delta$
algorithm \mytextsf{IS} outputs $v \in \R^n,$ $\theta \in \R$ which are such
that the LTF $h(x)=\sign(v \cdot x - \theta)$ has $\dshap(f,h) \leq \eps.$

\item For any input vector $(a(1),\dots, a(n))$, the
probability that \mytextsf{IS} outputs $v \in \R^n, \theta \in \R$
such that the LTF $h(x) =\sign(v \cdot x- \theta)$ has $\dshap(f,h) > 
\eps$ is at most $\delta.$

\end{enumerate}

\end{theorem}

\begin{proof}
We first note that we may assume $\eps > n^{-c}$ for a constant
$c>0$ of our choosing, for if $\eps \leq n^{-c}$ then the
claimed running time is $2^{\Omega(n^{2} \log n)}.$
In this much time we can easily enumerate all LTFs over $n$
variables (by trying all weight vectors with integer weights at most
$n^n$; this suffices by
\cite{MTT:61}) and compute their Shapley values exactly, and thus 
solve the problem.  So for the rest of the proof we assume
that $\eps > n^{-c}.$

It will be obvious from the description of \mytextsf{IS} that property
(2) above is satisfied, so the main job is to establish (1).
Before giving the formal proof we first describe an algorithm and analysis
achieving (1) for an idealized version of the problem.  We then
describe the actual algorithm and its analysis (which build on
the idealized version).

Recall that the algorithm is given as input  $\eps,\delta$ and $a(1),
\dots, a(n)$ that satisfy $\dshap(a,f) \leq
1/\poly(n,2^{\poly(1/\eps)})$ 
for some monotone increasing $\eta$-reasonable LTF $f$.
The idealized version of the problem is the following: we assume
that the algorithm is also given the two real values  
$f^*(0)$, $\sum_{i=1}^n {f}^*(i)/n$. It is also 
helpful to note that since $f$ is monotone and $\eta$-reasonable (and hence 
is not a constant function), it must be the case that 
$f(\mathbf{1})=1$ and $f(\mathbf{-1})=-1$. 

The algorithm for this idealized version is as follows:  
first, using Lemma~\ref{lem:expshap}, the 
values $\tilde{f}(i)$, $i=1,\dots,n$ are converted into values 
$a^*(i)$ which are approximations for the values
$f^*(i).$  Each $a^*(i)$ satisfies
$|a^*(i) - f^*(i)| \leq 1/\poly(n,2^{O(\poly(1/\eps))}).$
The algorithm sets $a^*(0)$ to $f^*(0).$
Next, the algorithm runs \mytextsf{Boosting-TTV} 
with the following input:  the family $\mathcal{L}$ of Boolean functions
is $\{1,x_1,\dots,x_n\}$; the values $a^*(0), \ldots, a^*(n)$ 
comprise the list of real values; $\mu$ is the distribution; and
the parameter $\xi$ is set to 
$1/\poly(n,2^{\poly(1/\eps)})$.
(We note that each execution of Step~3 of \mytextsf{Boosting-TTV},
namely finding values that closely estimate $\E_{x \sim \mu}[h_t(x)x_i]$
as required, is easily achieved using a standard sampling scheme;
for completeness 
in Appendix~\ref{sec:estimate}
we describe a procedure \mytextsf{Estimate-Correlation}
that can be used to do all the required estimations
with overall failure probability at most $\delta$.)
\mytextsf{Boosting-TTV} outputs an LBF $h(x) = P_1(v \cdot x
- \theta)$; the output of our overall algorithm is
the LTF $h'(x)=\sign(v \cdot x - \theta).$

Let us analyze this algorithm for the idealized scenario.
By Theorem~\ref{TTV}, the output function $h$ that 
is produced by \mytextsf{Boosting-TTV} is an LBF $h(x)=P_1(v \cdot x
- \theta)$ that satisfies
$\sqrt{\littlesum_{j=0}^n ({h}^*(j) - {f}^*(j))^2} = 1/\poly(n,2^{\poly(1/\eps)}).$
Given this, Lemma~\ref{lem:hatandstar} implies that
$\dfour(f,h) \leq 
\rho \eqdef 
1/\poly(n,2^{\poly(1/\eps)})$.

At this point, we have established that $h$ is a bounded function
that has $\dfour(f,h) \leq
1/\poly(n,2^{\poly(1/\eps)})$.  
We would like to apply
Lemma~\ref{lem:anticonc-and-ell1} and thereby assert that the
$\ell_1$ distance between $f$ and $h$ (with respect to $\mu$) is small.
To see that we can do this, we first note that 
since $f$ is a monotone increasing $\eta$-reasonable
LTF, by Theorem~\ref{thm:LParg} it has a representation as
$f(x)=\sign(w \cdot x + w_0)$ whose weights
satisfy the properties claimed in that theorem; in particular, for any choice of 
$\zeta>0$, 
after rescaling all the weights, the
largest-magnitude weight has magnitude 1, and 
the $k \eqdef \Theta_{\zeta,\eta}(1/\eps^{6 + 2 \zeta})$ largest-magnitude weights each have 
magnitude at least $r\eqdef 1/(n \cdot k^{O(k)})$.
(Note that since $\eps \geq n^{-c}$ we indeed have
$k \leq n$ as required.)
Given this, Theorem~\ref{thm:ac-general-weights} implies that 
the affine form $L(x) = w \cdot x + w_0$ satisfies
\begin{equation}
\Pr_{x \sim \mu}[|L(x)| < r] \leq \kappa \eqdef \eps^2/(\new{512}\log(n)),
\label{eq:ac1}
\end{equation}
i.e., it is $(r,\kappa)$-anticoncentrated with 
$\kappa=\eps^2/(\new{512}\log(n)).$
Thus we may indeed apply Lemma~\ref{lem:anticonc-and-ell1},
and it gives us that
\begin{equation}
\E_{x \sim \mu}[|f(x)-h(x)|] \leq {\frac {4 \|w\|_1 \sqrt{\rho}}
{r}} + \new{2} \kappa \leq \eps^2/(128 \log n).
\label{eq:zzz}
\end{equation}

Now let $h': \{-1,1\}^n \to \{-1,1\}$ be the LTF defined as
$h'(x)=\sign(v \cdot x - \theta)$ (recall that $h$ is the LBF
$P_1(v \cdot x - \theta)$).
Since $f$ is a $\{-1,1\}$-valued function, 
it is clear that for every input $x$ in the support of $\mu$,
the contribution of $x$ to 
$\Pr_{x \sim \mu}[f(x) \neq h'(x)]$  
is at most twice its contribution to 
$\E_{x \sim \mu}[|f(x)-h(x)|]$.
Thus we have that $\Pr_{x \sim \mu}[f(x) \neq h'(x)] \leq \eps^2/(64 \log n).$
We may now apply Fact~\ref{fac:func2}
to obtain that 
$\dfour(f,h') \leq \eps/(4\sqrt{\log n}).$
Finally, Lemma~\ref{lem:shapclose} gives that 
$$\dshap(f,h') \leq 4/\sqrt{n} + \sqrt{\Lambda(n)} \cdot \eps/(4 \sqrt{\log n})
< \eps/2.$$  
So indeed the LTF $h'(x) = \sign(v \cdot x - \theta)$
satisfies $\dshap(f,h') \leq \eps/ 2$ as desired.

\medskip

Now we turn from the idealized scenario to actually prove
Theorem~\ref{thm:main-arbitrary}, where we are not given the values of
${f}^*(0)$ and $\sum_{i=1}^n {f}^*(i)/n$. To get around this, we note that 
${f}^*(0),$ $\sum_{i=1}^n {f}^*(i)/n \in [-1,1].$
So the idea is that we will run the idealized algorithm repeatedly, trying ``all'' possibilities (up to some prescribed granularity) 
for ${f}^*(0)$ and for $\sum_{i=1}^n {f}^*(i)/n$. 
At the end of each such run we have a ``candidate'' LTF $h'$; we use a
simple procedure \mytextsf{Shapley-Estimate} 
(see Appendix~\ref{sec:estimate}) to estimate $\dshap(f,h')$
to within additive accuracy $\pm \eps/10$, and we output
any $h'$ whose estimated value of $\dshap(f,h')$ is at most $8\eps/10.$

We may run the idealized algorithm $\poly(n,2^{\poly(1/\eps)})$
times without changing its overall running time (up to
polynomial factors).  Thus we can try a net of possible
guesses for $f^*(0)$ and $\sum_{i=1}^n {f}^*(i)/n$
which is such that one guess will be within $\pm 1/
\poly(n,2^{\poly(1/\eps)})$ of the the correct values for
both parameters.
It is straightforward to verify
that the analysis of the idealized scenario given above is sufficiently
robust that when these ``good'' guesses are encountered, the algorithm will 
with high probability 
generate an LTF $h'$ that has $\dshap(f,h') \leq 6 \eps/10$.  
A straightforward analysis of running time and failure probability 
shows that properties (1) and (2) are achieved as desired,
and Theorem~\ref{thm:main-arbitrary} is proved.
\end{proof}

For any monotone $\eta$-reasonable target LTF $f$, 
Theorem~\ref{thm:main-arbitrary} 
constructs an output LTF whose Shapley distance from $f$ is at most $\eps$,
but the running time is exponential in $\poly(1/\eps).$
We now show that if the target monotone $\eta$-reasonable
LTF $f$ has integer weights that are at most $W$,
then we can construct an output LTF $h$ with $\dshap(f,h) \leq
n^{-1/8}$ running in time $\poly(n,W)$;
this is a far faster running time than
provided by Theorem~\ref{thm:main-arbitrary}
for such small $\eps$.  (The ``1/8'' 
is chosen for convenience; it will be clear from the proof that 
any constant strictly less than 1/6 would suffice.)

\ignore{

achieves $\eps$-accuracy for any 
Our second main result gives an algorithm that takes as input a weight bound for the desired
output LTF.  Its running time depends polynomially on this weight bound, but also only
polynomially on the accuracy parameter $\eps$ (compare with Theorem~\ref{thm:main-arbitrary}):

}

\begin{theorem} \label{thm:main-bounded}
There is an algorithm \mytextsf{ISBW} 
(for \mytextsf{Inverse-Shapley} with \mytextsf{Bounded Weights})
with the following properties.
\mytextsf{ISBW} is given as input a 
weight bound $W \in \Z_{+}$, a confidence parameter
 $\delta > 0$, and $n$ real values $a(1),\dots, a(n);$ its output is a pair
$v \in \R^n, \theta \in \R.$ 
Its running time is $\poly(n,W,\log(1/\delta)).$
The performance guarantees of \mytextsf{ISBW} are the following:

\begin{enumerate}

\item 
Suppose there is a monotone increasing $\eta$-reasonable 
LTF $f(x) = \sign(u \cdot x
- \theta)$, where each $u_i$ is an integer with $|u_i| \leq W$,
such that $\dshap(a,f) \leq 1/\poly(n,W).$  
Then with probability $1-\delta$
algorithm \mytextsf{ISBW} outputs $v \in \R^n,$ $\theta \in \R$ which are such
that the LTF $h(x)=\sign(v \cdot x - \theta)$ has $\dshap(f,h) \leq n^{-1/8}.$

\item For any input vector $(a(1),\dots,a(n))$, the
probability that \mytextsf{IS} outputs $v, \theta$
such that the LTF $h(x) =\sign(v \cdot x- \theta)$ has $\dshap(f,h) > 
n^{-1/8}$ is at most $\delta.$

\end{enumerate}

\end{theorem}

\begin{proof}
Let $f(x)=\sign(u \cdot x - \theta)$ be as described
in the theorem statement.  We may assume that each $|u_i| \geq 1$
(by scaling all the $u_i$'s and $\theta$ by $2n$ and then
replacing any zero-weight $u_i$ with 1).
Next we observe that for such an affine form $u \cdot x - \theta$,
Theorem~\ref{thm:ac-general-weights} immediately yields the 
following corollary:

\begin{corollary} \label{thm:ac-small-weights}
Let $L(x) = \littlesum_{i=1}^n u_i x_i - \theta$ be a monotone increasing 
$\eta$-reasonable affine form. 
Suppose that $u_i  \geq r$ for all $i=1,\dots,n.$ Then for any $\zeta>0$,
we have
\[ \Pr_{x \sim \mu} \left[| L(x)| <  r \right]  = O\left( \frac{1}{\log n}
\cdot \frac{1}{n^{1/3- \zeta}} \cdot \left( \frac{1}{\zeta} + \frac{1}{\eta}
\right)\right) .\]
\end{corollary}

With this anti-concentration statement in hand, the proof of 
Theorem~\ref{thm:main-bounded} closely follows the proof of
Theorem~\ref{thm:main-arbitrary}.
The algorithm runs \mytextsf{Boosting-TTV} with
${\cal L}$, $a^*(i)$ and $\mu$ as before but now with $\xi$
set to $1/\poly(n,W).$  The LBF $h$ that \mytextsf{Boosting-TTV} outputs
satisfies $\dfour(f,h) \leq \rho \eqdef 1/\poly(n,W)$.
We apply Corollary~\ref{thm:ac-small-weights} to the affine form
$L(x) \eqdef {\frac {u} {\|u\|_1}} \cdot x - {\frac {\theta}{\|u\|_1}}$
and get that for $r=1/\poly(n,W)$, we have
\begin{equation}
\Pr_{x \sim \mu}[|L(x)| < r] \leq \kappa \eqdef \eps^2/(1024
\log n)
\label{eq:redux}
\end{equation}
where now $\eps \eqdef n^{-1/8}$,
in place of Equation~(\ref{eq:ac1}).
Applying Lemma~\ref{lem:anticonc-and-ell1} we get that
\[
\E_{x \sim \mu}[|f(x)-h(x)|] \leq {\frac {4 \|w\|_1 \sqrt{\rho}}
{r}} + 4 \kappa \leq \eps^2/(128 \log n)
\]
analogous to (\ref{eq:zzz}).
The rest of the analysis goes through exactly as before, and we get
that the LTF $h'(x) = \sign(v \cdot x - \theta)$
satisfies $\dshap(f,h') \leq \eps/2$ as desired.  The rest
of the argument is unchanged so we do not repeat it.
\end{proof}

\new{

\section{Conclusions and Future Work} \label{sec:concl}

The problem of designing a weighted voting game
that (exactly or approximately) achieves a desired set of Shapley values has received considerable
attention in the social choice literature, where several heuristics and exponential time algorithms
have been proposed. This work provides the first provably correct efficient approximation algorithm for this problem. 

An obvious open problem is to improve the dependence on the error parameter $\eps$ in the running time.
Since the running time of our algorithm is of the form $\alpha(\eps) \cdot n^c$ for a fixed universal constant $c$, 
the algorithm is an Efficient Polynomial Time Approximation Scheme (EPTAS).
Is there a Fully Polynomial Time Approximation Scheme (FPTAS), i.e., an algorithm with running time
$\poly(n, 1/\eps)$? 

It would also be interesting to characterize the complexity of the {\em exact} problem (i.e., that of designing a weighted voting game
that {\em exactly} achieves a given set of Shapley values, or deciding that no such game exists). We conjecture that the exact problem
is intractable, namely $\sharp P$-hard.

}

\medskip

\noindent {\bf Acknowledgement.}  \new{We would like to thank Edith Elkind for asking the question about Shapley values and for useful pointers
to the literature.} We thank Christos  Papadimitriou for \new{insightful} conversations.

\bibliography{allrefs}
\bibliographystyle{alpha}

\newpage

\appendix

\section*{Appendix}

\section{LTF representations with ``nice" weights}\label{app:niceweights}

\ignore{
%\begin{definition}[$\gamma$-skewed affine form, LTF]
%An affine form $L(x) = w_0+ \littlesum_{i=1}^n w_i x_i$ is called {\em $\gamma$-skewed}, $\gamma \in (1/2,1)$, 
%if $|w_0| \leq (1-\gamma) \littlesum_{i=1}^n |w_i|.$ An LTF $f$ is called $\gamma$-skewed if it has a representation as $f(x) = \sgn(L(x))$
%for a $\gamma$-skewed affine form $L(x)$.
%\end{definition}

}
In this section, we prove Theorem~\ref{thm:LParg}.  This theorem
essentially says that given any $\eta$-reasonable LTF, there is an 
equivalent representation of the LTF which is also $\eta$-reasonable 
and is such that 
the weights of the linear form (when arranged in decreasing order of 
magnitude) decrease somewhat ``smoothly.''  For
convenience we recall the exact statement of the theorem:

\medskip

\noindent {\bf Theorem~\ref{thm:LParg}.}
\emph{
Let $f: \bn \to \bits$ be an $\eta$-reasonable LTF and $k \in [2, n]$. 
There exists a representation of $f$ as $f(x) = \sgn(v_0+ \littlesum_{i=1}^n 
v_i x_i)$ such that (after reordering \new{coordinates} so that condition (i) below 
holds) we have:   (i) $|v_{i}| \geq |v_{i+1}|$, $i \in [n-1]$; 
(ii) $|v_0| \leq (1-\eta) \littlesum_{i=1}^n |v_i|$; 
and (iii) for all $i \in [0, k-1]$ we have $|v_{i}| \leq (2/\eta) \cdot \sqrt{n} \cdot k^{\frac{k}{2}} \cdot \sigma_k$,
where $\sigma_k \eqdef \sqrt{\littlesum_{j \geq k} v_j^2}$.
}

\medskip

\begin{proof}[{\bf \em Proof of Theorem~\ref{thm:LParg}}]
The proof proceeds  along similar lines as the proof of Lemma~5.1 from \cite{OS11:chow} (itself an adaptation of the argument of Muroga et.~al.~from \cite{MTT:61}) 
with some crucial modifications. 

Since $f$ is $\eta$-reasonable, there exists a representation as $f(x) = \sgn(w_0 + \littlesum_{i=1}^n w_i x_i)$ 
(where we assume w.l.o.g. that $|w_i| \geq |w_{i+1}|$ for all $i \in [n-1]$) such that $|w_0| \leq (1-\eta) \littlesum_{i=1}^n |w_i|$.
Of course, this representation may not satisfy condition (iii) of the theorem statement.
We proceed to construct the desired alternate representation as follows: First, we set $v_i = w_i$ for all $i \geq k$. We then set up a feasible 
linear program $\mathcal{LP}$ with variables $u_0, \ldots, u_{k-1}$ and argue that there exists a feasible solution to 
$\mathcal{LP}$ with the desired properties.

Let $h: \{ \pm 1 \}^{k-1} \to \R$ denote the affine form $h(x) = w_0 + \littlesum_{j=1}^{k-1} w_j x_j$.
We consider the following linear system $\mathcal{S}$ of $2^{k-1}$ equations in $k$ unknowns $u_0, \ldots, u_{k-1}$:
For each $x \in  \{ \pm 1 \}^{k-1} $ we include the equation
$$ u_0+ \littlesum_{i=1}^{k-1} u_i x_i  = h(x).$$
It is clear that the system $\mathcal{S}$ is satisfiable, since $(u_0, \ldots, u_{k-1}) = (w_0, \ldots, w_{k-1})$ is a solution.

We now relax the above linear system into the linear program $\mathcal{LP}$ (over the same variables) as follows: Let $C \eqdef  \sqrt{n} \sigma_k$.
Our linear program has the following constraints:
\begin{itemize}
\item For each $x \in  \{ \pm 1 \}^{k-1} $ we include the (in)equality:
\begin{equation} \label{eqn:sensible-lp}
u_0+ \littlesum_{i=1}^{k-1} u_i x_i \
     \begin{cases}
        \geq C &  \text{if $h(x) \geq C$,} \\
        = h(x) &  \text{if $|h(x)| < C$,} \\
        \leq -C & \text{if $h(x) \leq -C$.} \\
     \end{cases}
\end{equation}

\item For each $i \in [0, k-1]$, we add the constraints $\sgn(u_i) = \sgn(w_i)$. Since the $w_i$'s are known, these are linear constraints, 
i.e., constraints like $u_1 \leq 0$, $u_2 \geq 0$, etc.

\item We also add the constraints of the form $|u_i| \geq |u_{i+1}|$ for $1\le i \le k-2$ and also $|u_{k-1}| \geq |w_{k}|$. Note that these constraints
are equivalent to the linear constraints: $u_i \cdot \sgn(w_i) \geq u_{i+1} \cdot \sgn(w_{i+1})$ and $\sgn(w_{k-1}) \cdot u_{k-1} \geq |w_{k}|$.

\item We let $q = \lceil 1/\eta \rceil$ and $\eta' = 1/q$. Clearly, $\eta' \le \eta$. 
We now add the constraint $|u_0| \leq (1-\eta') \cdot \left(\littlesum_{j=1}^{k-1} |u_j| + \littlesum_{j=k}^n |w_j| \right)$.
Note that this is also a linear constraint over the variables $u_0, u_1, \ldots, u_{k-1}.$ Indeed, it can be equivalently written as:
$$  \sgn(w_0) \cdot u_0 - (1-\eta') \littlesum_{j=1}^{k-1} \sgn(w_j) \cdot  u_j  \leq (1-\eta') \littlesum_{j=k}^n |w_j|.$$ 
Note that the RHS is \new{strictly} bounded from above by $C$, \new{since $$ \littlesum_{j=k}^n |w_j| \le \sqrt{n-k+1} \cdot \sigma_k < \sqrt{n} \sigma_k,$$
\noindent where the first inequality is Cauchy-Schwarz and the second uses the fact that $k \ge 2$.}
\end{itemize}

We observe that the above linear program is feasible. Indeed, it is straightforward to verify that all the constraints are satisfied by the 
vector $(w_0, \ldots, w_{k-1}).$  In particular, the last constraint is satisfied because $|w_0| \le (1-\eta) \cdot \left(\littlesum_{j=1}^{k-1} |w_j| + \littlesum_{j=k}^n |w_j| \right)$ and hence \emph{a fortiori},   $|w_0| \le (1-\eta') \cdot \left(\littlesum_{j=1}^{k-1} |w_j| + \littlesum_{j=k}^n |w_j| \right)$.
\begin{myclaim}
Let $(v_0, \ldots, v_{k-1})$ be {\em any} feasible solution to $\mathcal{LP}$ and consider the LTF 
$$f'(x) = \sgn(v_0+ \littlesum_{j=1}^{k-1}v_{j}x_j +  \littlesum_{j=k}^{n} w_{j}x_j  ).$$ Then 
$f'(x) = f(x)$
for all $x \in \bn$.
\end{myclaim}
\begin{proof}
Given $x \in \bn$, we have
$$h(x) = h(x_1, \ldots, x_{k-1}) = w_0 + \littlesum_{j=1}^{k-1} w_jx_j;$$ 
Let us also define $$h'(x) = h'(x_1, \ldots, x_{k-1}) = v_0 + \littlesum_{j=1}^{k-1} v_jx_j$$
$$t(x) = \littlesum_{j \geq k} w_j x_j$$
Then, we have $f(x) = \sgn \left( h(x)+t(x) \right)$ and $f'(x) = \sgn \left( h'(x)+t(x) \right)$. Now, if $x \in \bn$ is an input such that $|h(x)| < C$, then we have
$h'(x) = h(x)$ by construction, and hence $f(x) = f'(x)$. If $x \in \bn$ is such that $|h(x)| \geq C$, then by construction we also have that $|h'(x)| \geq C$.
Also, note that $h(x)$ and $h'(x)$ always have the same sign. Hence, in order for $f$ and $f'$ to disagree on $x$, it must be the case that $|t(x)| \geq C$.
But this is not possible, since $|t(x)| \leq \littlesum_{j=k}^n |w_j| \le  \new{\sqrt{n-1} \cdot} \sigma_k < C.$ This completes the proof of the claim.
\end{proof}

We are almost done, except that we need to choose a 
solution $(v_0, \ldots, v_{k-1})$ to $\mathcal{LP}$ satisfying 
property (iii) in the statement of the theorem. The next claim ensures that 
this can always be achieved. 
\begin{myclaim}\label{clm:bounds}
There is a feasible solution  $v = (v_0, \ldots, v_{k-1})$ to the $\mathcal{LP}$ which satisfies property (iii) in the statement of the theorem.
\end{myclaim}
\begin{proof}
We select a feasible solution $v = (v_0, \ldots, v_{k-1})$ to the $\mathcal{LP}$ that maximizes the number of {\em tight} inequalities (i.e., satisfied with equality).
If more than one feasible solutions satisfy this property, we choose one 
arbitrarily. We require the following fact from \cite{MTT:61} (a 
proof can be found in  \cite{Hastad:94,DiakonikolasServedio:09short}).
\begin{fact}
There exists a linear system $A \cdot v = b$ that uniquely specifies the vector $v$. The rows of $(A, b)$ correspond to rows of the constraint matrix of $\mathcal{LP}$
and the corresponding RHS respectively.
\end{fact}

\medskip

At this point, we use Cramer's rule to complete the argument. In particular, note that $v_i = \det(A_i)/\det (A)$ where $A_i$ is the matrix obtained by replacing the $i$-th column of $A$ by $b$. In particular, we want to give an upper bound on the magnitude of $v_i$; we do this by showing a lower bound on $|\det (A)|$ and an upper bound on $|\det(A_i)|$. 

We start by showing that $|\det(A)| \ge  \eta'$.  First, since $A$ is invertible, $\det(A) \neq 0$. Now, note that all rows of $A$ have entries in $\{-1,0, 1\}$ except potentially one ``special'' row which has entries from the set $\{\pm 1, \pm (1-\eta') \}$.  If the special row does not appear, it is clear that $|\det(A)| \geq 1$, since it is not zero and the entries of $A$ are all integers. If, on the other hand, the special row appears, simply expanding $\det (A)$ along that row gives that $\det (A) = a \cdot (1- \eta') + b$ where $a, b \in \Z$. 
As $\eta' = 1/q$ for some $q \in Z$ and $\det(A) \neq 0$, we deduce that $|\det(A)| \ge \eta'$, as desired. 

We bound  $|\det(A_i)|$ from above by recalling the following fact. 
\begin{fact}(Hadamard's inequality) 
If $A \in \R^{n \times n}$ and $v_1, \ldots, v_n \in \R^n$  are the columns of $A$, then $|\det(A)| \le \prod_{j=1}^n \Vert v_j \Vert_2$. 
\end{fact}
Now, observe that for all $i$, the $i$-th column of $A_i$ (i.e., vector $b$) has all its entries bounded by $C$, hence  $\Vert v_i \Vert_2 \le C \sqrt{k}$. All other columns have entries bounded from above by $1$ and thus for $j \neq i$, $\Vert v_j \Vert_2 \le \sqrt{k}$.  Therefore, $\det(A_i) \le C \cdot k^{k/2}$. 
Thus, we conclude that $|v_i| \le (C \cdot k^{k/2})/\eta'.$ Further, as $(1/\eta')  = \lceil (1/\eta) \rceil \le (2/\eta)$, we get $|v_i| \le 2C \cdot k^{k/2}/\eta$, completing the proof of the claim. 
%\qed
%At this point, we can use Cramer's rule to complete the argument. Each variable $v_i$ can be expressed as the ratio of two determinants.
%The determinant in the denominator is $\det(A)$. It is easy to show that $|\det(A)| \geq \min \{ \eta, 1-\eta\}$;  it's value is certainly non-zero, 
%since the $A$ except potentially one have entries in $\{-1, 0,1\}$. 
%The row $[1, 1-\eta, \ldots, 1-\eta]$ may appear in $A$. (I.e. the row corresponding to last constraint of the linear program.) 
%This row gives us the bound on the value of the corresponding determinant  (by expanding). 
%(The bound on the value of this determinant one gets this way seems tight to me, if we don't assume sth more. In particular, we get that 
%$\det(A) = m_1 + (1-\eta)m_2$, where $m_1, m_2 \in \Z$. Not both of $m_1, m_2$ can be $0$, since the matrix is invertible.
%value of $1-\eta$; if $m_1 = 1$ and $m_2 = -1$ we get value of $\eta$. These are the smallest possible values one can get.)
%As for the numerator, it is the determinant of the matrix obtained from $A$ by replacing it's $i$-th column by the RHS vector $b$; 
%hence, it is at most $Ck^{k}$, by Hadamard's inequality. This gives as the desired upper bound on the values of the $v_i$'s.
%(Since we assumed that $\eta$ is at least $1/2$ or sth.)
\end{proof}
The proof of Theorem~\ref{thm:LParg} is now complete.
\end{proof}

\section{Estimating correlations and Shapley values}
\label{sec:estimate}

Our algorithms need to estimate expectations of the form
$f^*(i)=\E_{x \sim \mu}[f(x)x_i]$ and to estimate Shapley values 
$\tilde{f}(i)$, where
$f: \{-1,1\}^n \to [-1,1]$ is an explicitly given function (an LBF). This is 
quite straightforward using standard techniques (see e.g. \cite{BMRPRS10}) 
but for completeness we briefly state and prove the estimation guarantees 
that we will need.

\medskip
 
\noindent {\bf Estimating correlations with variables.}  We will use
the following:

\begin{proposition}\label{obs:corr}
There is a procedure \mytextsf{Estimate-Correlation} with the following 
properties:  The procedure is given oracle access to a function 
$f : \{-1,1\}^n \rightarrow [-1,1]$, a desired accuracy parameter $\gamma$, 
and a desired failure probability $\delta$.  The procedure makes 
$O(n \log (n/\delta) /\gamma^2)$ oracle calls to $f$ and runs in 
time $O(n^2  \log (n/\delta) /\gamma^2)$ (counting
each oracle call to $f$ as taking one time step).
With probability $1-\delta$ it outputs a list of numbers ${a}^*(0),a^*(1),
\ldots, {a}^*(n)$ such that $|a^*(j) - f^*(j)| \le \gamma/\sqrt{n+1}$ for 
all $j=0,\dots,n.$ (Recall that $f^*(j)$ equals $\E_{x \sim \mu}[f(x)x_j]$, 
where $x_0 \equiv 1$).
\end{proposition}
\begin{proof}
The procedure works simply by empirically estimating
all the values $f^*(j)=\E_{x \sim \mu}[f(x)x_j]$, $j=0,\dots,n,$ 
using a single sample of
$m$ independent draws from $\mu$.  Since the random variable 
$(f(x)x_j))_{x \sim \mu}$ is bounded by 1 in absolute value, a 
straightforward Chernoff bound gives that for $m=O(n \log(n/\delta)/\gamma^2),$
each estimate $a^*(j)$ of $f^*(j)$ is accurate to within
an additive $\pm \gamma/\sqrt{n+1}$ with failure probability
at most $\delta/(n+1)$.  A union bound over $j=0,\dots,n$ finishes
the argument.
\end{proof}

\noindent {\bf Estimating Shapley values.}  This is equally straightforward:

\begin{proposition}\label{obs:shapley}
There is a procedure \mytextsf{Estimate-Shapley} with the following properties:
The procedure is given oracle access to a function 
$f : \{-1,1\}^n \rightarrow [-1,1]$, a desired accuracy parameter $\gamma$, 
and a desired failure probability $\delta$.  The procedure makes 
$O(n \log (n/\delta) /\gamma^2)$ oracle calls to $f$ and runs in 
time $O(n^2  \log (n/\delta) /\gamma^2)$ (counting
each oracle call to $f$ as taking one time step).
With probability $1-\delta$ it outputs a list of numbers $\tilde{a}(1),
\ldots, \tilde{a}(n)$ such that $\dshap(a,f) \leq \gamma.$
\end{proposition}
\begin{proof}
The procedure empirically estimates each $\tilde{f}(j)$, $j=1,\dots,n$,
to additive accuracy $\gamma/\sqrt{n}$ using Equation~(\ref{eq:shapley}).
This is done by generating a uniform random $\pi \sim \ms$ and then,
for each $i=1,\dots,n,$
constructing the two inputs $x^+(\pi,i)$ and $x(\pi,i)$ and calling
the oracle for $f$ twice to compute $f(x^+(\pi,i))-f(x(\pi,i)).$
Since $|f(x^+(\pi,i))-f(x(\pi,i))| \leq 2$ always, a
sample of $m=O(n \log(n/\delta)/\gamma^2)$ permutations suffices
to estimate all the $\tilde{f}(i)$ values to additive
accuracy $\pm \gamma/\sqrt{n}$ with total failure probability
at most $\delta$.  If each estimate $\tilde{a}(i)$ is additively
accurate to within $\pm \gamma/\sqrt{n}$, then $\dshap(a,f)
\leq \gamma$ as desired.
\end{proof}

% %%%%%%%%%%%%%%%%%%%%%%%%%%%%%%%%%%%%%%%%%%%%%%%%%%%%%%%%%%%%%%%%

% %%%%%%%%%%%%%%%%%%%%%%%%%%%%%%%%%%%%%%%%%%%%%%%%%%%%%%%%%%%%%%%%%%%%%%%%%
% End of appendix 
% %%%%%%%%%%%%%%%%%%%%%%%%%%%%%%%%%%%%%%%%%%%%%%%%%%%%%%%%%%%%%%%%%%%%%%%%%

\end{document}